 \documentclass[twocolumn,twoside]{IEEEtran}
\usepackage{amsmath,amssymb,amsthm,wasysym,epsfig,color,subfigure,graphicx,epstopdf,url,bm,nicefrac,tcolorbox}
\usepackage{caption}
\usepackage{amsmath,graphicx,bm,accents}
\usepackage{algorithm}
\usepackage{algorithmic}
\allowdisplaybreaks

\usepackage{hyperref}
\hypersetup{colorlinks=false}

\usepackage{multicol}
\usepackage{multirow}

\usepackage{accents}
\makeatletter
\def\wid{\check{{\cc@style\underline{\mskip9.5mu}}}}
\def\Wideubar{\underaccent{{\cc@style\underline{\mskip6mu}}}}
\makeatother

\makeatletter
\def\wideubar{\underaccent{{\cc@style\underline{\mskip9.5mu}}}}
\def\Wideubar{\underaccent{{\cc@style\underline{\mskip6mu}}}}
\makeatother

\makeatletter
\def\widebar{\accentset{{\cc@style\underline{\mskip9.5mu}}}}
\def\Widebar{\accentset{{\cc@style\underline{\mskip6mu}}}}
\makeatother

\newtheorem{theorem}{Theorem}

\theoremstyle{remark}\newtheorem{remark}{Remark}

\begin{document}
\title{Graph Multiview Canonical Correlation Analysis}

\author{Jia Chen,
	Gang Wang,~\IEEEmembership{Member,~IEEE},
	and 
	Georgios B. Giannakis,~\IEEEmembership{Fellow,~IEEE}
	\thanks{This work was supported in part by NSF grants 1500713, 1514056, 1505970, and 1711471. 
	The authors are with the Digital Technology Center and the Department of Electrical and Computer Engineering, University of Minnesota, Minneapolis, MN 55455, USA. 
E-mails: \{chen5625,\, gangwang, \,georgios\}@umn.edu.
		}
}

\maketitle

\allowdisplaybreaks

\begin{abstract}
Multiview canonical correlation analysis (MCCA) seeks latent low-dimensional  representations encountered with multiview data of shared entities (a.k.a. common sources). However, existing MCCA approaches do not exploit the geometry of the common sources, which may be available \emph{a priori}, or can be constructed using certain domain knowledge. This prior information about the common sources can be encoded by a graph, and be invoked as a regularizer to enrich the maximum variance MCCA framework. In this context, the present paper's novel graph-regularized (G) MCCA  approach minimizes the distance between the wanted canonical variables and the common low-dimensional representations, while accounting for graph-induced knowledge of the common sources. Relying on a function capturing the extent low-dimensional representations of the multiple views are similar, a generalization bound of GMCCA is established based on Rademacher's complexity. Tailored for setups where the number of data pairs is smaller than the data vector dimensions, a graph-regularized dual MCCA approach is also developed. To further deal with nonlinearities present in the data, graph-regularized kernel MCCA variants are put forward too. Interestingly, solutions of the graph-regularized linear, dual, and kernel MCCA, are all provided in terms of generalized eigenvalue decomposition. Several corroborating numerical tests using real datasets are provided to showcase the merits of the graph-regularized MCCA variants relative to several competing alternatives including MCCA, Laplacian-regularized MCCA, and (graph-regularized) PCA. 
\end{abstract}

\begin{IEEEkeywords}
	Dimensionality reduction, canonical correlation analysis, signal processing over graphs, Laplacian regularization, generalized eigen-decomposition, multiview learning
\end{IEEEkeywords}

\section{Introduction}
\label{sec:intro}
In several applications, such as multi-sensor surveillance systems, multiple datasets are collected offering distinct views of the common information sources. With advances in data acquisition, it becomes easier to access heterogeneous data representing samples from multiple views in various scientific fields, including genetics, computer vision, data mining, and pattern recognition, to name a few. In genomics for instance, a patient's lymphoma data set consists of gene expression, SNP, and array CGH measurements~\cite{witten2009gene}. In a journal's dataset, the title, keywords, and citations can be considered as three different views of a given paper~\cite{tang2009clustering}. Learning with heterogeneous data of different types is commonly referred to as multiview learning, and in different communities as information fusion or data integration from multiple feature sets. Multiview learning is an emerging field in data science with well-appreciated analytical tools and matching application domains~\cite{2013mlsurvey}.  

Canonical correlation analysis (CCA) is a classical tool for multiview learning \cite{1936cca}. Formally, CCA looks for latent low-dimensional representations from a paired dataset comprising two views of several common entities. Multiview (M) CCA generalizes two-view CCA and also principal component analysis (PCA) \cite{1901pca}, to handle jointly datasets from multiple views \cite{1971kettenringcca}. In contrast to PCA that operates on vectors formed by multi-view sub-vectors, MCCA is more robust to outliers per view, because it ignores the principal components per view that are irrelevant to the latent common sources. Popular MCCA formulations include the sum of correlations  (SUMCOR), maximum variance  (MAXVAR) \cite{1961maxvar}, sum of squared correlations, the minimum variance, and generalized variance methods \cite{1971kettenringcca}. With the increasing capacity of data acquisition and the growing demand for multiview data analytics, the research on MCCA has been re-gaining attention recently.

To capture nonlinear relationships in the data, linear MCCA has been also generalized using (multi-)kernels or deep neural networks; see e.g., \cite{2003kcca,andrew2013deep,wang2015deep}, that have well-documented merits for (nonlinear) dimensionality reduction of multiview data, as well as for multiview feature extraction. Recent research efforts have also focused on addressing the scalability issues in (kernel)  MCCA, using random Fourier features \cite{lopez2014randomized}, or leveraging alternating optimization advances \cite{kanatsoulis2018structured} to account for sparsity
\cite{witten2009penalized,2014smcca,chen2017distributed,kanatsoulis2018structured} or other types of structure-promoting regularizers such as nonnegativity and smoothness~\cite{chen2012structured,2015mccanonneg}.

Lately, graph-aware regularizers have demonstrated promising performance in a gamut of machine learning applications, such as dimensionality reduction, data reconstruction, clustering, and classification
\cite{gpca1, gpca,jstsp2016shahid,2012gdmf,proc2018gsk,2018cwsggcca}. CCA with structural information induced by a common source graph has been reported in 
\cite{2018cwsggcca}, but it is limited to analyzing two-views of data, and its performance has been tested only experimentally. Further, multigraph-encoded information provided by the underlying physics, or, inferred from alternative views of the information sources, has not been investigated.

Building on but considerably going beyond our precursor work in \cite{2018cwsggcca}, this paper introduces a novel graph-regularized (G) MCCA approach, and develops a bound on its generalization error performance. Our GMCCA is established by minimizing the difference between the low-dimensional representation of each view and the common representation, while also leveraging the statistical dependencies due to the common sources hidden in the views. These dependencies are encoded by a graph, which can be available from the given data, or can be deduced from correlations. A finite-sample statistical analysis of GMCCA is provided based on a regression formulation offering a meaningful error bound for unseen data samples using Rademacher's complexity. 

GMCCA is operational when there are sufficient data samples (larger than the number of features per view). For cases where the data are insufficient, we develop a graph-regularized dual (GD) MCCA scheme that avoids this limitation at lower computational complexity. To cope with nonlinearities present in real data, we further put forward a graph-regularized kernel (GK) MCCA scheme. Interestingly, the linear, dual, and kernel versions of our proposed GMCCA admit simple analytical-form solutions, each of which can be obtained by performing a single generalized eigenvalue decomposition.

Different from \cite{blaschko2011semi, pr2014sun}, where MCCA is regularized using multiple graph Laplacians separately per view, GMCCA here jointly leverages a single graph effected on the common sources. This is of major practical importance, e.g., in electric power networks, where besides the power, voltage, and current quantities observed, the system operator has also access to the network topology \cite{pssechap} that captures the connectivity between substations through power lines. 

Finally, our proposed GMCCA approaches are numerically tested using several real datasets on different machine learning tasks, including e.g., dimensionality reduction, recommendation, clustering, and classification. Corroborating tests showcase the merits of GMCCA schemes relative to its completing alternatives such as MCCA, PCA, graph PCA, and the k-nearest neighbors (KNN) method. 

\emph{Notation}: Bold uppercase (lowercase) letters denote matrices (column vectors). Operators ${\rm Tr}(\cdot)$, $(\cdot)^{-1}$, ${\rm vec}(\cdot)$ and $(\cdot)^{\top}$ stand for matrix trace, inverse, vectorization, and transpose, respectively; $\|\cdot\|_2$ denotes the $\ell_2$-norm of vectors; $\|\cdot\|_F$ the Frobenius norm of matrices;  ${\rm diag}(\{a_m\}_{m=1}^M)$ is an $M\times M$ diagonal matrix holding entries of $\{a_m\}_{m=1}^M$ on its main diagonal; $ \langle \mathbf{a},\, \mathbf{b}\rangle$ denotes the inner product of same-size vectors $\mathbf{a}$ and $\mathbf{b}$; vector $\mathbf{0}$ has all zero entries whose dimension is clear from the context; and $\mathbf{I}$ is the identity matrix of suitable size.

\section{Preliminaries}
\label{sec:review}
Consider  $M$ datasets $\{\mathbf{X}_m\in\mathbb{R}^{D_m\times N}\}_{m=1}^M$ collected from $M\ge 2$ views of $N$ common source vectors $\{\check{\mathbf{s}}_n\in\mathbb{R}^\rho  \}_{n=1}^N$ stacked as columns of $\check{\mathbf{S}}\in\mathbb{R}^{\rho\times N} $,
where $D_m$ is the dimension of the $m$-th view data vectors, with possibly $\rho \ll \min_m \, \{D_m\}_{m=1}^M $. Vector $\mathbf{x}_{m,i}$ denotes the $i$-th column of $\mathbf{X}_m$, meaning the $i$-th datum of the $m$-th view,
for all $i=1,\ldots,N$ and $m=1,\ldots,M$. Suppose without loss of generality that all per-view data vectors $\{\mathbf{x}_{m,i}\}_{i=1}^N$ have been centered. 
Two-view CCA works with datasets $\{\mathbf{x}_{1,i} \}_{i=1}^N$ and $\{\mathbf{x}_{2,i} \}_{i=1}^N$ from $M=2$ views. It looks for low-dimensional subspaces $\mathbf{U}_1\in\mathbb{R}^{D_1\times d}$ and $\mathbf{U}_2\in\mathbb{R}^{D_2\times d}$ with $d\le \rho$, such that the Euclidean distance between linear projections $\mathbf{U}_1^\top\mathbf{X}_1$ and $\mathbf{U}_2^\top\mathbf{X}_2$ is minimized. Concretely, classical CCA solves the following problem \cite{hardoon2004canonical}
\begin{subequations}
	\label{eq:cca}
	\begin{align}
	\underset{\mathbf{U}_1,\mathbf{U}_2}{	\min}\quad&\! \left\|\mathbf{U}_1^\top\mathbf{X}_1-\mathbf{U}_2^\top\mathbf{X}_2\right \|_F^2\label{eq:ccacos}\\
		{\rm s.\,to}\;\quad& \mathbf{U}_m^\top\!\left(\mathbf{X}_m^\top\mathbf{X}_m\right)\mathbf{U}_m=\mathbf{I},\quad  m=1,\,2\label{eq:ccacon}
	\end{align}
\end{subequations} 
where columns of $\mathbf{U}_m$ are called loading vectors of the data (view) $\mathbf{X}_m$; while projections $\{\mathbf{U}_m^\top\mathbf{X}_m\}_{m=1}^2$ are termed canonical variables; they satisfy \eqref{eq:ccacon} to prevent the trivial solution; and, they can be viewed as low ($d$)-dimensional approximations of $\check{\mathbf{S}}$. Moreover, the solution of \eqref{eq:cca} is provided by a generalized eigenvalue decomposition \cite{1936cca}. 

When analyzing multiple ($\ge 3$) datasets, \eqref{eq:cca} can be generalized to a pairwise matching criterion
\cite{ccasumcor}; that is
\begin{subequations}
		\label{eq:mccapw}
	\begin{align}
\min_{\{\mathbf{U}_m\}_{m=1}^M}~ &\! \sum_{m=1}^{M-1}\sum_{m'>m}^{M}\! \left\|\mathbf{U}_m^\top\mathbf{X}_m-\mathbf{U}_{m'}^\top\mathbf{X}_{m'}\right\|_F^2\label{eq:mccapwcos}\\
	{\rm s. \,to}\,\,\quad&\mathbf{U}_m^\top\!\left(\mathbf{X}_m^\top\mathbf{X}_m\right)\mathbf{U}_m=\mathbf{I},\quad m=1,\ldots,M\label{eq:mccapwcon}
	\end{align}
\end{subequations}
where \eqref{eq:mccapwcon} ensures a unique nontrivial solution. The formulation 
in \eqref{eq:mccapw} is referred to as the sum-of-correlations (SUMCOR) MCCA, that  is known to be NP-hard in general \cite{ccanphard}. 

Instead of minimizing the distance between paired low-dimensional approximations, one can look for a shared low-dimensional representation of different views, namely $\mathbf{S}\in\mathbb{R}^{d\times N}$, by solving  \cite{1971kettenringcca}
\begin{subequations}
	\label{eq:mccacs}
	\begin{align}
	\min_{\{\mathbf{U}_m\}_{m=1}^M,\mathbf{S}}  \quad& \sum_{m=1}^M \left \|\mathbf{U}_m^\top\mathbf{X}_m-\mathbf{S}\right\|_F^2\label{eq:mccacscos}\\
	{\rm s. \,to}\,\,\,~~ \quad&\; \mathbf{S}\mathbf{S}^\top=\mathbf{I}\label{eq:mccacscon}
	\end{align}
\end{subequations}
yielding the so-called  maximum-variance (MAXVAR) MCCA formulation.
Similarly, the constraint \eqref{eq:mccacscon} is imposed to avoid a trivial solution. If all per-view sample covariance matrices $\{\mathbf{X}_m\mathbf{X}^\top_m\}_m$ have full rank, then for a fixed  $\mathbf{S}$, the $\mathbf{U}_m$-minimizers are given by  $\{\hat{\mathbf{U}}_m=(\mathbf{X}_m\mathbf{X}_m^\top)^{-1}\mathbf{X}_m\mathbf{S}^\top\}_m$. Substituting $\{\hat{\mathbf{U}}_m \}_m$ into \eqref{eq:mccacs}, the $\mathbf{S}$-minimizer can be obtained by solving the following eigenvalue decomposition problem
	\begin{subequations}
	\label{eq:mccacss}
	\begin{align}
\hat{\mathbf{S}}:=\arg	\underset{\mathbf{S}}\max ~~&{\rm Tr}\Big[\mathbf{S}\Big(\sum_{m=1}^M\mathbf{X}_m^\top\left (\mathbf{X}_m\mathbf{X}_m^\top\right)^{-1}\mathbf{X}_m\Big)\mathbf{S}^\top\Big]\\
{\rm s. \,to}  ~~&\, \mathbf{S}\mathbf{S}^\top=\mathbf{I}.
	\end{align}
	\end{subequations}
The columns of $\hat{\mathbf{S}}^\top$ are given by the first $d$ principal eigenvectors of matrix $\sum_{m=1}^M\mathbf{X}_m^\top(\mathbf{X}_m\mathbf{X}_m^\top)^{-1}\mathbf{X}_m$.
In turn, we deduce that  $\{\hat{\mathbf{U}}_m=(\mathbf{X}_m\mathbf{X}_m^\top)^{-1}\mathbf{X}_m\hat{\mathbf{S}}^\top\}_{m=1}^M$.

A couple of comments are worth noting about \eqref{eq:mccacs} and \eqref{eq:mccacss}. 
\begin{remark} Solutions of the SUMCOR MCCA in \eqref{eq:mccapw} and the MAXVAR MCCA in \eqref{eq:mccacs} are generally different. Specifically, for $M=2$, both admit analytical solutions that can be expressed in terms of distinct eigenvalue decompositions; but for $M > 2$, the SUMCOR MCCA can not be solved analytically, while the MAXVAR MCCA still admits an analytical solution though at the price of higher computational complexity because it involves the extra matrix variable $\bf{S}$.  
\end{remark}

\section{Graph-regularized MCCA}
In many applications, the common source vectors $\{\check{\mathbf{s}}_i\}_{i=1}^N$ may reside on, or their dependencies form a graph of $N$ nodes. This structural prior information can be leveraged along with multiview datasets to improve MCCA performance. Specifically, we will capture this extra knowledge here using a graph, and effect it in the low-dimensional common source estimates through a graph regularization term.

Consider representing the graph of the $N$ common sources using the tuple $\mathcal{G}:=\{\mathcal{N},\,\mathcal{W}\}$, where $\mathcal{N}:=\{1,\ldots,N\}$ is the vertex set, and $\mathcal{W}:=\{w_{ij}\}_{(i,j)\in\mathcal{N}\times \mathcal{N}}$ collects all edge weights $\{w_{ij}\}$ over all vertex pairs $(i,\,j)$. The so-termed weighted adjacency matrix $\mathbf{W}\in\mathbb{R}^{N\times N}$ is formed with $w_{ij}$ being its $(i,\,j)$-th entry. Without loss of generality, undirected graphs for which $\mathbf{W}=\mathbf{W}^\top$ holds are considered in this work. Upon defining $d_i:=\sum_{j=1}^N w_{ij}$ and $\mathbf{D}:={\rm diag}(\{d_i\}_{i=1}^N)\in\mathbb{R}^{N\times N}$, the Laplacian matrix of graph $\mathcal{G}$ is defined as
\begin{equation}
	\label{eq:lg}
	\mathbf{L}_{\mathcal{G}}:=\mathbf{D}-\mathbf{W}.
\end{equation}

Next, a neat link between canonical correlations and graph regularization will be elaborated. To start, let us assume that sources $\{\check{\mathbf{s}}_i\}_{i=1}^N$ are smooth over $\mathcal{G}$. This means that two sources $(\check{\mathbf{s}}_i,\,\check{\mathbf{s}}_j)$ residing on two connected nodes $i,\,j\in\mathcal{N}$ are also close to each other in Euclidean distance. As explained before, 
vectors $\mathbf{s}_i$ and $\mathbf{s}_j$
are accordingly the $d$-dimensional approximations of $\check{\mathbf{s}}_i$ and $\check{\mathbf{s}}_j$. Accounting for this fact, a meaningful regularizer is the weighted sum of distances between any pair of common source estimates $\mathbf{s}_i$ and $\mathbf{s}_j$ over $\mathcal{G}$
\begin{equation}
\label{eq:g}
{\rm Tr}\big(\mathbf{S}\mathbf{L}_{\mathcal{G}}\mathbf{S}^{\top}\big)=\sum_{i=1}^N\sum_{j=1}^N w_{ij}\!\left \|\mathbf{s}_i-\mathbf{s}_j\right \|_2^2.
\end{equation}
Clearly, 
source vectors $\mathbf{s}_i$ and $\mathbf{s}_j$ residing on adjacent nodes $i,\,j\in\mathcal{N}$ having large weights $w_{ij}$ will be forced to be similar to each other. To leverage such additional graph information of the common sources, the quadratic term \eqref{eq:g} is invoked as a regularizer in the standard MAXVAR MCCA, 
yielding our novel graph-regularized (G) MCCA formulation
\begin{subequations}
	\label{eq:gmcca}
	\begin{align}
	\underset{\{\mathbf{U}_m\}\atop\mathbf{S}}\min\quad &\sum_{m=1}^M\!\left \|\mathbf{U}_m^\top\mathbf{X}_m-\mathbf{S}\right\|_F^2+\gamma{\rm Tr}\! \left(\mathbf{S}\mathbf{L}_{\mathcal{G}}\mathbf{S}^\top\right )\label{eq:gmccacos}\\
	{\rm s. \,to}\,\quad&~\mathbf{S}\mathbf{S}^\top=\mathbf{I}
	\end{align}
\end{subequations}
where the coefficient $\gamma\ge 0$ trades off minimizing the distance between the canonical variables and their corresponding common source estimates with promoting smoothness of common source estimates over the graph $\mathcal{G}$.
Specifically, when $\gamma=0$,  GMCCA reduces to the classical MCCA in \eqref{eq:mccacs}; and, as $\gamma$ increases, GMCCA relies more heavily in this extra graph knowledge when finding the canonical variables. 

If all per-view sample covariance matrices $ \{\mathbf{X}_m\mathbf{X}_m^\top\}$ have full rank, equating to zero the partial derivative of the cost in \eqref{eq:gmccacos} with respect to each  $\mathbf{U}_m$, yields the optimizer $\hat{\mathbf{U}}_m=(\mathbf{X}_m\mathbf{X}_m^\top)^{-1}\mathbf{X}_m\mathbf{S}^\top$.
Substituting next $\mathbf{U}_m$ by $\hat{\mathbf{U}}_m$ and ignoring the constant term in \eqref{eq:gmccacos} give rise to the following eigenvalue problem (cf. \eqref{eq:mccacss})
\begin{subequations}
	\label{eq:gmccas}
	\begin{align}
	\underset{\mathbf{S}}\max\quad&{\rm Tr}\Big[\mathbf{S}\Big(\sum_{m=1}^M\mathbf{X}_m^\top\big(\mathbf{X}_m\mathbf{X}_m^\top\big)^{-1}\mathbf{X}_m-\gamma\mathbf{L}_{\mathcal{G}}\Big)\mathbf{S}^\top\Big]\label{eq:gmccascos}\\
	{\rm s. \,to}  \quad& \mathbf{S}\mathbf{S}^\top=\mathbf{I}.
	\end{align}
\end{subequations}
Similar to standard MCCA, the optimal solution $\hat{\bf{S}}$ of \eqref{eq:gmccas} can be obtained by the $d$ leading eigenvectors of the matrix  
\begin{equation}
\label{eq:cmatrix}
\mathbf{C}:=\sum_{m=1}^M\mathbf{X}_m^\top(\mathbf{X}_m\mathbf{X}_m^\top)^{-1}\mathbf{X}_m-\gamma\mathbf{L}_{\mathcal{G}}.
\end{equation}
At the optimum, it is easy to verify that the following holds 
	\begin{align*}
	\label{eq:optimalobj}
	\sum_{m=1}^M\Big  \|\hat{\mathbf{U}}_m^\top\mathbf{X}_m-\hat{\mathbf{S}}\Big\|_F^2+\gamma{\rm Tr} \big(\hat{\mathbf{S}}\mathbf{L}_{\mathcal{G}}\hat{\mathbf{S}}^\top\big )=Md-\sum_{i=1}^d \lambda_i
	\end{align*}
where $\lambda_i$ denotes the $i$-th largest eigenvalue of $\mathbf{C}$ in \eqref{eq:cmatrix}.

A step-by-step description of our proposed GMCCA scheme is summarized in Alg. \ref{alg:gmcca}. 

At this point, a few remarks are in order.
\begin{algorithm}[t]
	\caption{Graph-regularized MCCA.}
	\label{alg:gmcca}
	\begin{algorithmic}[1]
		\STATE {\bfseries Input:} $\{\mathbf{X}_m\}_{m=1}^M$, $d$, $\gamma$, and $\mathbf{W}$.
		\STATE {\bfseries Build} $\mathbf{L}_{\mathcal{G}}$ using \eqref{eq:lg}.
		\STATE {\bfseries Construct}  $\mathbf{C}=\!\!\sum_{m=1}^M\!\mathbf{X}_m^\top\left(\mathbf{X}_m\mathbf{X}_m^\top\right)^{-1}\mathbf{X}_m-\gamma\mathbf{L}_{\mathcal{G}}$.
		\STATE {\bfseries Perform} \label{step:4} eigendecomposition
		on $\mathbf{C}$ to obtain the $d$ eigenvectors associated with the $d$ largest eigenvalues, which are collected as columns of $\hat{\mathbf{S}}^\top$.
		\STATE{\bfseries Compute} $\big\{\hat{\mathbf{U}}_m=\left(\mathbf{X}_m\mathbf{X}_m^\top\right)^{-1}\mathbf{X}_m\hat{\mathbf{S}}^\top\big\}_{m=1}^M$.
		\STATE {\bfseries Output:} $\{\hat{\mathbf{U}}_m\}_{m=1}^M$ and $\hat{\mathbf{S}}$.
		\vspace{-0pt}
	\end{algorithmic}
\end{algorithm}

\begin{remark}\label{rmk:compwold}
We introduced a two-view graph CCA scheme in \cite{2018cwsggcca} using the SUMCOR MCCA formulation. However, to obtain an analytical solution, the original cost was surrogated in \cite{2018cwsggcca} by its lower bound, which cannot be readily generalized for multiview datasets with $M\ge 3$. In contrast, our GMCCA  in \eqref{eq:gmcca} can afford an analytical solution for any $M\ge 2$.
\end{remark}
\begin{remark}
\label{rmk:graphcca}
Different from our single graph regularizer in \eqref{eq:gmcca}, the proposals in \cite{blaschko2011semi} and \cite{pr2014sun} rely on $M$ different regularizers $\{\mathbf{U}_m^\top\mathbf{X}_m\mathbf{L}_{\mathcal{G}_m}\mathbf{X}_m^\top\mathbf{U}_m\}_m$ to exploit the extra graph knowledge, for view-specific graphs $\{\mathbf{L}_{\mathcal{G}_m}\}_m$ on data $\{\mathbf{X}_m\}_m$. However, the formulation in \cite{pr2014sun} does not admit an analytical solution, and convergence of the iterative solvers for the resulting nonconvex problem can be guaranteed only to a stationary point. The approach in \cite{blaschko2011semi} focuses on semi-supervised learning tasks, in which cross-covariances of pair-wise datasets are not fully available. In contrast, the single graph Laplacian regularizer in \eqref{eq:gmcca} is effected on the common sources, to exploit the pair-wise similarities of the $N$ common sources. This is of practical importance when one has prior knowledge about the common sources besides the $M$ datasets. For example, in ResearchIndex networks, besides keywords, titles, Abstracts, and Introductions of collected articles, one has also access to the citation network capturing the connectivities among those papers. More generally, the graph of inter-dependent sources can be dictated by underlying physics, or it can be a prior provided by an `expert,' or, it can be learned from extra (e.g., historical) views of the data. Furthermore, our proposed GMCCA approach comes with simple analytical solutions.
\end{remark}

\begin{remark}\label{rmk:gamma}
With regards to selecting $\gamma$, two ways are feasible: i) cross-validation for supervised learning tasks, where labeled training data are given, and $\gamma$ is fixed to the one that yields optimal empirical performance on the training data; and, ii) using a spectral clustering method that automatically chooses the best $\gamma$ values from a given set of candidates; see e.g., \cite{tsp2018cwg}.
\end{remark}

\begin{remark}\label{rmk:scale}
Our GMCCA scheme entails eigendecomposition of an $N\times N$ matrix, which incurs computational complexity $\mathcal{O}(N^3)$, and thus is not scalable to large datasets. Possible remedies include parallelization and efficient decentralized 
algorithms capable of handling structured MCCA; e.g., along the lines of \cite{kanatsoulis2018structured}. These go beyond the scope of the present paper, but constitute interesting future research directions. 
\end{remark}

\section{Generalization Bound of GMCCA}\label{sec:gb}
In this section, we will analyze the finite-sample performance of GMCCA based on a regression formulation \cite[Ch. 6.5]{shawe2004kernel}, which is further related to the alternating conditional expectations method in \cite{1985ace}. Our analysis will establish an error bound for unseen source vectors (a.k.a. generalization bound) using the notion of Rademacher's complexity.

Recall that the goal of MCCA is to find common low-dimensional representations of the $M$-view data.
To measure how close the estimated $M$ low-dimensional representations are to each other, we introduce the following
error function
\begin{equation}\label{eq:gs}
g(\check{\mathbf{s}}):=\sum_{m=1}^{M-1}
\sum_{m'>m}^{M}\!\left\|\mathbf{U}_m^\top\bm{\psi}_m(\check{\mathbf{s}})-\mathbf{U}_{m'}^\top\bm{\psi}_{m'}(\check{\mathbf{s}})\right\|_F^2
\end{equation}
where the underlying source vector $\check{\mathbf{s}}\in\mathbb{R}^{\rho}$ is assumed to follow some fixed yet unknown distribution $\mathcal{D}$, and the linear function $\bm{\psi}_m(\cdot)$ maps a source vector from space $\mathbb{R}^{\rho}$ to the $m$-the view in $\mathbb{R}^{D_m}$, for $m=1,\ldots,M$. 

To derive the generalization bound, we start by evaluating the empirical average of $g(\check{\mathbf{s}})$ over say, a number $N$ of given training samples, as follows
\begin{align*}
&\bar{g}_N(\check{\mathbf{s}}):=\frac{1}{N}\!\sum_{n=1}^N \sum_{m=1}^{M-1}\!\sum_{m'>m}^{M}\!\left \|\mathbf{U}_m^\top\bm{\psi}_m(\check{\mathbf{s}}_n)-\mathbf{U}_{m'}^\top\bm{\psi}_{m'}(\check{\mathbf{s}}_n)\right\|_F^2\\
&=\frac{1}{N}\!\sum_{n=1}^N \sum_{m=1}^{M-1}\! \sum_{m'>m}^{M}\!\Big[\bm{\psi}_m^\top\!(\check{\mathbf{s}}_n) \mathbf{U}_m \mathbf{U}_m^\top \bm{\psi}_m(\check{\mathbf{s}}_n) -2\bm{\psi}_m^\top \!(\check{\mathbf{s}}_n)\\
&\quad \times \mathbf{U}_m\mathbf{U}_{m'}^\top \bm{\psi}_{m'}(\check{\mathbf{s}}_n)
\!+ \bm{\psi}_{m'}^\top\! (\check{\mathbf{s}}_n)\mathbf{U}_{m'} \!\mathbf{U}_{m'}^\top \bm{\psi}_{m'}(\check{\mathbf{s}}_n)\Big].
\end{align*}
For the quadratic terms, it can be readily verified that 
\begin{align}
\label{eq:vec}
\bm{\psi}_m^\top (\check{\mathbf{s}}) &\mathbf{U}_m \mathbf{U}_m^\top \bm{\psi}_m(\check{\mathbf{s}})=\left\langle {\rm vec} (\mathbf{U}_m\mathbf{U}_m^\top),\,{\rm vec}  (\bm{\psi}_m(\check{\mathbf{s}}) \bm{\psi}_m^\top (\check{\mathbf{s}}) ) \right\rangle\\
\label{eq:vecij}
\bm{\psi}_m^\top (\check{\mathbf{s}})
&\mathbf{U}_m \mathbf{U}_{m'}^\top \bm{\psi}_{m'}(\check{\mathbf{s}})\!=\!\left\langle {\rm vec} (\mathbf{U}_m\mathbf{U}_{m'}^\top),{\rm vec}(\bm{\psi}_m(\check{\mathbf{s}}) \bm{\psi}_{m'}^\top(\check{\mathbf{s}}))\! \right\rangle.
\end{align}

Define two $\sum_{m=1}^{M-1}\sum_{m'>m}^M(D_m^2\!+\!D_{m'}^2\!+\!D_mD_{m'})\times 1$ vectors 
\begin{align*}
\bm{\psi}(\check{\mathbf{s}})&:= \! \left [\bm{\psi}^\top_{11}(\check{\mathbf{s}})~\cdots~\bm{\psi}^\top_{1M}(\check{\mathbf{s}})~\bm{\psi}^\top_{23}(\check{\mathbf{s}})~\cdots~\bm{\psi}^\top_{M,M-1}(\check{\mathbf{s}}) \right]^\top\\
\mathbf{u} &: = \left[\mathbf{u}_{11}^\top~\cdots~\mathbf{u}_{1M}^\top~\mathbf{u}_{23}^\top~\cdots~\mathbf{u}_{M,M-1}^\top \right]^\top
\end{align*}
where the two $(D_m^2+D_{m'}^2+D_mD_{m'})\times 1$ vectors $\bm{\psi}_{mm'}(\check{\mathbf{s}})$ and $\mathbf{u}_{mm'}$ are defined as
	\begin{align*}
	&	\bm{\psi}_{mm'}:=\big[{\rm vec}^\top(\bm{\psi}_m\bm{\psi}_m^\top)~
	{\rm vec}^\top(\bm{\psi}_{m'}\bm{\psi}_{m'}^\top)	~
	\sqrt{2}{\rm vec}^\top(\bm{\psi}_m\bm{\psi}_{m'}^\top)
	\big]^\top\\
	&	\mathbf{u}_{mm'}\!:=\!\big[{\rm vec\!}^\top\!(\mathbf{U}_m\mathbf{U}_m^\top)~ {\rm vec\!}^\top\!(\mathbf{U}_{m'}\mathbf{U}_{m'}^\top)~\!-\!\sqrt{2}{\rm vec\!}^\top\!(\mathbf{U}_m\mathbf{U}_{m'}^\top)\big]\!^\top
	\end{align*}
for $m=1,\ldots,M-1$ and $m'=2,\ldots,M$.

Plugging \eqref{eq:vec} and \eqref{eq:vecij} into \eqref{eq:gs}, one can check that function $g(\check{\mathbf{s}})$ can be rewritten as 
\begin{equation}
\label{eq:gupsi}
g(\check{\mathbf{s}})=\left\langle \mathbf{u},\, \bm{\psi}(\check{\mathbf{s}}) \right\rangle.
\end{equation}
with the norm of $\mathbf{u}$ given by
\begin{equation*}
\|\mathbf{u}\|_2^2=\sum_{m=1}^{M-1}\sum_{m'>m}^M\left\|\mathbf{U}_m^\top\mathbf{U}_m+\mathbf{U}_{m'}^\top\mathbf{U}_{m'}\right\|_F^2.
\end{equation*}

Starting from \eqref{eq:gupsi}, we will establish next an upperbound on the expectation of $g(\check{\mathbf{s}})$ by means of \eqref{eq:gupsi}, which is important because the expectation involves not only the $N$ training source samples, but also unseen samples.

\begin{theorem}\label{thm:geb}
	Assume that i) the $N$ common source vectors $\{\check{\mathbf{s}}_n\}_{n=1}^N$ are drawn i.i.d. from some distribution $\mathcal{D}$; ii) the $M$ transformations
	$\{\bm{\psi}_m(\cdot)\}_{m=1}^M$ of vectors $\{\check{\mathbf{s}}_n\}_{n=1}^N$ are bounded; and, iii) subspaces $\{\mathbf{U}_m\in\mathbb{R}^{D_m\times d}\}_{m=1}^M$ satisfy  $\sum_{m=1}^{M-1}\sum_{m'>m}^M\|\mathbf{U}_m^\top\mathbf{U}_m+\mathbf{U}_{m'}^\top\mathbf{U}_{m'}\|_F^2\le B^2$ ($B>0$) and $\{\mathbf{U}_m\}_{m=1}^M$ are the optimizers of \eqref{eq:gmcca}.	If we obtain low-dimensional representations of $\{\bm{\psi}_m(\check{\mathbf{s}})\}_{m=1}^{M}$ specified by subspaces $\{\mathbf{U}_m\in\mathbb{R}^{D_m\times d}\}_{m=1}^M$, it holds with probability at least $1-\delta$ that
	\begin{align}
	&\mathbb{E}[g(\check{\mathbf{s}})] \le\bar{g}_N(\check{\mathbf{s}}) + 3RB\sqrt{\frac{{\rm ln}(2/\delta)}{2N}}\nonumber\\
	&~+\frac{4B}{N}\sqrt{\sum_{n=1}^N\sum_{m=1}^{M-1}\sum_{m'>m}^M \left[\kappa_m(\check{\mathbf{s}}_n,\,\check{\mathbf{s}}_n) +\kappa_{m'}(\check{\mathbf{s}}_n,\,\check{\mathbf{s}}_n)\right]^2 }\label{eq:thm}
	\end{align}
	where 
$
	\kappa_m(\check{\mathbf{s}}_n,\check{\mathbf{s}}_n):=\left\langle\bm{\psi}_m(\check{\mathbf{s}}_n),\bm{\psi}_m(\check{\mathbf{s}}_n)\right\rangle
$
	for $n=1,\ldots,N$, and $m=1,\ldots,M$, while the constant $R$ is given by
	\begin{equation*}
	R:=\underset{\check{\mathbf{s}}\sim \mathcal{D}}\max\sqrt{\sum_{m=1}^{M-1}\sum_{m'>m}^M \left[\kappa_m(\check{\mathbf{s}},\check{\mathbf{s}})+\kappa_{m'}(\check{\mathbf{s}},\check{\mathbf{s}})\right]^2 }.
	\end{equation*}
	
\end{theorem}
\begin{proof}
	Equation \eqref{eq:gupsi} suggests that $g(\check{\mathbf{s}})$ belongs to the function class
	\begin{equation*}
	\mathcal{F}_B:=\left\{\check{\mathbf{s}}\to \left\langle \mathbf{u},\,\bm{\psi}(\check{\mathbf{s}})\right\rangle: \|\mathbf{u}\|\le B\right\}.
	\end{equation*}
	Consider the function class
	\begin{equation*}
	\mathcal{H}=\left\{h: \check{\mathbf{s}}\to 1/(RB) f(\check{\mathbf{s}})\big|f(\cdot)\in\mathcal{F}_B\right\}\subseteq \mathcal{A}\circ\mathcal{F}_B
	\end{equation*}
	where the function $\mathcal{A}$ is defined as
	\begin{equation*}
	\label{eq:a}
	\mathcal{A}(x)=\left\{
	\begin{array}{ccl}
	0, & & {{\rm if}~ x\le 0}\\
	\frac{x}{RB}, & & {{\rm if}~ 0\le x\le RB}\\
	1, & & {{\rm otherwise}}
	\end{array} \right..
	\end{equation*}
		
It can be checked that $\mathcal{A}(\cdot)$ is a Lipschitz function with Lipschitz constant $1/(RB)$, and that the range of functions in $\mathcal{H}$ is $[0,\,1]$. Appealing to \cite[Th. 4.9]{shawe2004kernel}, one deduces that with probability at least $1-\delta$, the following holds
	\begin{align}
	\mathbb{E}[h(\check{\mathbf{s}})]&\le \frac{1}{N}\sum_{n=1}^N h(\mathbf{s}_n )+R_N(\mathcal{H})+\sqrt{\frac{{\rm ln 2/\delta}}{2N}}\nonumber\\
	&\le  \frac{1}{N}\sum_{n=1}^N h(\check{\mathbf{s}}_n )+\hat{R}_N(\mathcal{H})+3\sqrt{\frac{{\rm ln 2/\delta}}{2N}}\label{eq:gebh}
	\end{align}	
	where $\mathbb{E}[h(\check{\mathbf{s}})]$ denotes the expected value of $h(\cdot)$ on a new common source $\check{\mathbf{s}}$; and the Rademacher complexity $R_N(\mathcal{H})$ of $\mathcal{H}$ along with its empirical version $\hat{R}_N(\mathcal{H})$ is defined as
	\begin{align*}
	&	R_N(\mathcal{H}):=\mathbb{E}_{\check{\mathbf{s}}}[\hat{R}_N(\mathcal{H})]\\
		&\hat{R}_N(\mathcal{H}):=\mathbb{E}_{\bm{\delta}} \Big [\underset{h\in\mathcal{H}}{\rm sup}\Big  |\frac{2}{N}\sum_{n=1}^N \delta_n h(\check{\mathbf{s}}_n)\Big |\left  |\check{\mathbf{s}}_1,\,\check{\mathbf{s}}_2,\,\ldots,\,\check{\mathbf{s}}_N\right.\Big]
	\end{align*}
where $\bm{\delta}:=\{\delta_n\}_{n=1}^N$ collects independent random variables drawn from the Rademacher distribution, meaning $\{{\rm Pr}(\delta_n=1)={\rm Pr}(\delta_n=-1)=0.5\}_{n=1}^N$. Further, $\mathbb{E}_{\bm{\delta}}[\cdot]$ and $\mathbb{E}_{\check{\mathbf{s}}}[\cdot]$ denote the expectation with respect to  $\bm{\delta}$ and $\check{\mathbf{s}}$, respectively.
	
	Since $\mathcal{A}(\cdot)$ is a Lipschitz function with Lipschitz constant $1/(RB)$ satisfying $\mathcal{A}(0)=0$, the result in \cite[Th. 12]{bartlett2002rademacher} asserts that
	\begin{equation}\label{eq:hrc}
	\hat{R}_N({\mathcal{H}})\le 2/(RB)\hat{R}_N(\mathcal{F}_B).
	\end{equation} 
	Applying \cite[Th. 4.12]{shawe2004kernel} leads to 
	\begin{equation}\label{eq:fbrc}
	\hat{R}_N(\mathcal{F}_B)\le 2B/N\sqrt{{\rm Tr}(\mathbf{K})}
	\end{equation}
	where the $(i,\,j)$-th entry of $\mathbf{K}\in\mathbb{R}^{N\times N}$ is $\big\langle\bm{\psi}(\check{\mathbf{s}}_i),\bm{\psi}(\check{\mathbf{s}}_j)\big\rangle$, for $i,\,j=1,\ldots,N$. One can also confirm that 
	\begin{equation}\label{eq:tracek}
	\!\!{\rm Tr}(\mathbf{K})=\sum_{n=1}^N\sum_{m=1}^{M-1}\sum_{m'>m}^M \!\!\,\Big[\kappa_m(\check{\mathbf{s}}_n,\check{\mathbf{s}}_n) +\kappa_{m'}(\check{\mathbf{s}}_n,\check{\mathbf{s}}_n)\Big]^2.
	\end{equation}
	Substituting \eqref{eq:fbrc} and \eqref{eq:tracek} to \eqref{eq:hrc} yields
	\begin{equation*}\label{eq:h2rc}
	\hat{R}_N({\mathcal{H}})\!\le\! \frac{4}{RN} \sqrt{\!\sum_{n=1}^N\!\sum_{m=1}^{M-1}\!\sum_{m'>m}^M \!\!\Big[\kappa_m(\check{\mathbf{s}}_n,\check{\mathbf{s}}_n) +\kappa_{m'}(\check{\mathbf{s}}_n,\check{\mathbf{s}}_n)\Big]^2 }.
	\end{equation*}
	Multiplying \eqref{eq:gebh} by $RB$ along with the last equation
	gives rise to \eqref{eq:thm}.	
\end{proof}

Theorem \ref{thm:geb} confirms that the empirical expectation of $g(\cdot)$, namely $\bar{g}_N(\check{\mathbf{s}})$, stays close to its ensemble one $\mathbb{E}(g(\check{\mathbf{s}}))$, provided that $\{\|\mathbf{U}_m\|_F\}_{m}$ can be controlled. For this reason, it is prudent to trade off maximization of correlations among the $M$ datasets with the norms of the resultant loading vectors.
\section{Graph-regularized Dual MCCA}
In practical scenarios involving high-dimensional data vectors with dimensions satisfying $\min_m D_m > N$, the matrices $\{\mathbf{X}_m\mathbf{X}_m^\top\}$ become singular -- a case where GMCCA in \eqref{eq:gmcca} does not apply. For such cases, consider rewriting the $D_m\times d$ loading matrices $\mathbf{U}_m$ in terms of the data matrices $\mathbf{X}_m$ as $\mathbf{U}_m =\mathbf{X}_m\mathbf{A}_m$, where  $\mathbf{A}_m\in\mathbb{R}^{N\times d}$ will be henceforth termed the dual of $\mathbf{U}_m$. Replacing $\mathbf{U}_m$ with $\mathbf{X}_m\mathbf{A}_m$ in the linear GMCCA formulation \eqref{eq:gmcca} leads to its dual formulation
\begin{subequations}
	\label{eq:gmccad}
	\begin{align}
	\min_{\{\mathbf{A}_m\},\,\mathbf{S}}~&~\sum_{m=1}^M\left\|\mathbf{A}_m^\top\mathbf{X}_m^\top\mathbf{X}_m-\mathbf{S}\right\|_F^2+\gamma{\rm Tr}\left(\mathbf{S}\mathbf{L}_{\mathcal{G}}\mathbf{S}^\top\right)\label{eq:gmccadcos}\\
	{\rm s. ~to}\quad&~~\mathbf{S}\mathbf{S}^\top=\mathbf{I}.\label{eq:gmccadcon}
	\end{align}
\end{subequations}
If the $N\times N$ matrices $\{\mathbf{X}_m^\top\mathbf{X}_m\}_{m=1}^M$ are nonsingular, it can be readily confirmed that the $d\leq \rho$ columns of the optimizer $\hat{\mathbf{S}}^\top$ of \eqref{eq:gmccad} are the $d$ principal eigenvectors of $M\mathbf{I}-\gamma\mathbf{L}_{\mathcal{G}}$, while the dual matrices can be estimated in closed form as  $\hat{\mathbf{A}}_m=(\mathbf{X}_m^\top\mathbf{X}_m)^{-1}\hat{\mathbf{S}}^\top$.
Clearly, such an $\hat{\mathbf{S}}$ does not depend on the data $\{\mathbf{X}_m\}_{m=1}^M$, and this estimate goes against our goal of extracting $\hat{\mathbf{S}}$ as the latent low-dimensional structure commonly present in $\{\mathbf{X}_m\}_{m=1}^M$. To address this issue, we mimic the dual CCA trick (see e.g., \cite{hardoon2004canonical}), and introduce a Tikhonov regularization term on the loading vectors through the norms of $\left\{\|\mathbf{U}_m\|_F^2={\rm Tr}\left(\mathbf{A}_m^\top\mathbf{X}_m^\top\mathbf{X}_m\mathbf{A}_m\right)\right\}$. This indeed agrees with the observation we made following Theorem \ref{thm:geb} that controlling $\{\|\mathbf{U}_m\|_F^2\}$ improves the generalization. In a nutshell, our graph-regularized dual (GD) MCCA is given as
\begin{subequations}
	\label{eq:gmccadr}
	\begin{align}
	\underset{\{\mathbf{A}_m\},\mathbf{S}}\min~&~~\sum_{m=1}^M\left\|\mathbf{A}_m^\top\mathbf{X}_m^\top\mathbf{X}_m-\mathbf{S}\right\|_F^2+\gamma{\rm Tr}\left(\mathbf{S}\mathbf{L}_{\mathcal{G}}\mathbf{S}^\top\right)\nonumber\\
	& ~+\sum_{m=1}^M \epsilon_m {\rm Tr}\left(\mathbf{A}_m^\top\mathbf{X}_m^\top\mathbf{X}_m\mathbf{A}_m\right)\label{eq:gmccadrcos}\\
	{\rm s. ~to}~~&~~\mathbf{S}\mathbf{S}^\top=\mathbf{I}.\label{eq:gmccadrcon}
	\end{align}
\end{subequations}
where $\{\epsilon_m\ge 0\}$ denote pre-selected weight coefficients. 

As far as the solution is concerned, it can be deduced that the $i$-th column of the optimizer $\hat{\mathbf{S}}$ of \eqref{eq:gmccadr} is the eigenvector of $\mathbf{C}_d:=\sum_{m=1}^M (\mathbf{X}_m^\top\mathbf{X}_m + \epsilon \mathbf{I})^{-1}-\gamma\mathbf{L}_{\mathcal{G}}$ associated with the $i$-th largest eigenvalue. Once $\hat{\mathbf{S}}$ is found, the optimal dual matrices can be obtained as $\{\hat{\mathbf{A}}_m=(\mathbf{X}_m^\top\mathbf{X}_m+\epsilon\mathbf{I})^{-1}\hat{\mathbf{S}}^\top\}_{m=1}^M$. The steps of implementing GDMCCA are summarized in Alg. \ref{alg:gdmcca}.

\begin{algorithm}[t]
	\caption{Graph-regularized  dual MCCA.}
	\label{alg:gdmcca}
	\begin{algorithmic}[1]
		\STATE {\bfseries Input:} $\{\mathbf{X}_m\}_{m=1}^M$, $\epsilon$, $\gamma$, and $\mathbf{W}$.
		\STATE {\bfseries Build} $\mathbf{L}_{\mathcal{G}}$ using \eqref{eq:lg}.
		\STATE {\bfseries Construct} $\mathbf{C}_d=\sum_{m=1}^M \left(\mathbf{X}_m^\top\mathbf{X}_m + \epsilon \mathbf{I}\right)^{-1}-\gamma\mathbf{L}_{\mathcal{G}}$.
		\STATE {\bfseries Perform} \label{step:4} eigenvalue decomposition
		on $\mathbf{C}_d$ to obtain the $d$ eigenvectors associated with the $d$ largest eigenvalues, which are collected as columns of $\hat{\mathbf{S}}^\top$.
		\STATE{\bfseries Compute} $\{\hat{\mathbf{A}}_m=\left(\mathbf{X}_m^\top\mathbf{X}_m+\epsilon\mathbf{I}\right)^{-1}\hat{\mathbf{S}}^\top\}_{m=1}^M$.
		\STATE {\bfseries Output:} $\{\hat{\mathbf{A}}_m\}_{m=1}^M$ and $\hat{\mathbf{S}}$.
		\vspace{-0pt}
	\end{algorithmic}
\end{algorithm} 

\section{Graph-regularized Kernel MCCA}
The GMCCA and GDMCCA approaches are limited to analyzing linear data dependencies. Nonetheless, complex nonlinear data dependencies are not rare in practice. To account for nonlinear dependencies, a graph-regularized kernel (GK) MCCA formulation is pursued in this section to capture the nonlinear relationships in the $M$ datasets $\{\mathbf{X}_m\}_m$ through kernel-based methods. Specifically, the idea of GKMCCA involves first mapping the data vectors $\{\mathbf{X}_m \}_m$ to higher (possibly infinite) dimensional feature vectors by means of $M$ nonlinear functions, on which features we will apply GMCCA to find the shared low-dimensional canonical variables. 

Let $\bm{\phi}_m$ be a mapping from $\mathbb{R}^{D_m}$ to $\mathbb{R}^{L_m}$ for all $m$, where the dimension $L_m$ can be as high as infinity. Clearly, the data enter the GDMCCA problem \eqref{eq:gmccadr} only via the similarity matrix $\mathbf{X}_m^\top\mathbf{X}_m$. Upon mapping all data vectors $\{\mathbf{x}_{m,i}\}_{i=1}^N$ into $\{\bm{\phi}_m(\mathbf{x}_{m,i})\}_{i=1}^N$, the linear similarities $\{\langle \mathbf{x}_{m,i},\,\mathbf{x}_{m,j}\rangle\}_{i,j=1}^N$ can be replaced with the mapped nonlinear similarities $\{\langle \bm{\phi}_m(\mathbf{x}_{m,i}),\, \bm{\phi}_m(\mathbf{x}_{m,j})\rangle \}_{i,j=1}^N$. After selecting some kernel function $\kappa^m$ such that $\kappa^m(\mathbf{x}_{m,i},\, \mathbf{x}_{m,j}):=\langle \bm{\phi}_m(\mathbf{x}_{m,i}), \, \bm{\phi}_m(\mathbf{x}_{m,j})\rangle$, the $(i,\, j)$-th entry of the kernel matrix $\bar{\mathbf{K}}_m\in\mathbb{R}^{N\times N}$ is given by $\kappa^m(\mathbf{x}_{m,i},\, \mathbf{x}_{m,j})$, for all $i$, $j$, and $m$. In the sequel, centering $\{\bm{\phi}_m(\mathbf{x}_{m,i})\}_{i=1}^N$ is realized by centering the kernel matrix for data $\mathbf{X}_{m}$ as
\begin{align}
\label{eq:km}
\mathbf{K}_m(i,\,j):=&\,\bar{\mathbf{K}}_m(i,j)-\frac{1}{N}\sum_{k=1}^N \bar{\mathbf{K}}_m(k,j)-\frac{1}{N} \sum_{k=1}^N \bar{\mathbf{K}}_m(i,k)\nonumber\\
&+\frac{1}{N^2} \sum_{i,j=1}^N \bar{\mathbf{K}}_m(i,j)
\end{align}
for $m=1,\ldots,M$. 

Replacing $\{\mathbf{X}_m^\top\mathbf{X}_m\}_m$ in the GDMCCA formulation \eqref{eq:gmccadr} with centered kernel matrices $\{\mathbf{K}_m\}_m$ yields our GKMCCA
\begin{subequations}
	\label{eq:gmccak}
	\begin{align}
	\underset{\{\mathbf{A}_m\},\mathbf{S}}\min&\quad \sum_{m=1}^M\left\|\mathbf{A}_m^\top\mathbf{K}_m-\mathbf{S}\right\|_F^2+\gamma{\rm Tr}\left(\mathbf{S}\mathbf{L}_{\mathcal{G}}\mathbf{S}^\top\right)\nonumber\\
	&\quad +\epsilon\sum_{m=1}^M  {\rm Tr}\left(\mathbf{A}_m^\top\mathbf{K}_m\mathbf{A}_m\right)\label{eq:gmccakcos}\\
	{\rm s. ~to}~&\quad \mathbf{S}\mathbf{S}^\top=\mathbf{I}.\label{eq:gmccakcon}
	\end{align}
\end{subequations} 

Selecting invertible matrices $\{\mathbf{K}_m\}_{m=1}^M$, and following the logic used to solve \eqref{eq:gmccadr}, we can likewise tackle \eqref{eq:gmccak}.  Consequently, the columns of the optimizer $\hat{\mathbf{S}}^\top$ are the first $d$ principal eigenvectors of $\mathbf{C}_g:= \sum_{m=1}^M  (\mathbf{K}_m+\epsilon \mathbf{I})^{-1}\mathbf{K}_m - \gamma\mathbf{L}_{\mathcal{G}} \in\mathbb{R}^{N\times N}$, and the optimal  $\hat{\mathbf{A}}_m$ sought can be obtained as $\hat{\mathbf{A}}_m=(\mathbf{K}_m+\epsilon\mathbf{I})^{-1}\hat{\mathbf{S}}^\top$. For implementation, GKMCCA is presented in step-by-step form as Algorithm \ref{alg:gkmcca}.
\begin{algorithm}[t]
	\caption{Graph-regularized kernel MCCA.}
	\label{alg:gkmcca}
	\begin{algorithmic}[1]
		\STATE {\bfseries Input:} $\{\mathbf{X}_m\}_{m=1}^M$, $\epsilon$, $\gamma$, $\mathbf{W}$, and $\{\kappa^m\}_{m=1}^M$.
		\STATE {\bfseries Construct} $\{\mathbf{K}_m\}_{m=1}^M$ using \eqref{eq:km}.
		\STATE {\bfseries Build} $\mathbf{L}_{\mathcal{G}}$ using \eqref{eq:lg}.
		\STATE {\bfseries Form} $\mathbf{C}_g= \sum_{m=1}^M  \left(\mathbf{K}_m+\epsilon \mathbf{I}\right)^{-1}\mathbf{K}_m - \gamma\mathbf{L}_{\mathcal{G}}$.
		\STATE {\bfseries Perform} \label{step:4} eigendecomposition
		on $\mathbf{C}_g$ to obtain the $d$ eigenvectors associated with the $d$ largest eigenvalues, which are collected as columns of $\hat{\mathbf{S}}^\top$.
		\STATE{\bfseries Compute} $\{\hat{\mathbf{A}}_m=\left(\mathbf{K}_m+\epsilon\mathbf{I}\right)^{-1}\hat{\mathbf{S}}^\top\}_{m=1}^M$.
		\STATE {\bfseries Output:} $\{\hat{\mathbf{A}}_m\}_{m=1}^M$ and $\hat{\mathbf{S}}$.
		\vspace{-0pt}
	\end{algorithmic}
\end{algorithm} 

In terms of computational complexity, recall that GMCCA, GDMCCA, GKMCCA, MCCA, DMCCA, and KMCCA all require finding the eigenvectors of matrices with different dimensionalities. Defining $D:=\max_{m} D_m$, it can be checked that 
they incur correspondingly complexities $\mathcal{O}(N^2{\rm max}(N,DM))$, $\mathcal{O}(N^2DM)$, $\mathcal{O}(N^2M\max(N,D))$, $\mathcal{O}(N^2{\rm max}(N,DM))$, $\mathcal{O}(N^2DM)$, and $\mathcal{O}(N^2M\max(N,D))$. 
Interestingly, introducing graph-regularization to e.g., MCCA, DMCCA, as well as KMCCA does not result in an increase of computational complexity. When $\{N\ll D_m\}_{m=1}^M$, GMCCA in its present form is not feasible, or suboptimal even though pseudo-inverse can be utilized at the cost of $\mathcal{O}(MD^3)$. In contrast, GDMCCA is computationally preferable as its cost grows only linearly with $D$. When $N\gg D$, the complexity of GKMCCA is dominated by the computation burden of $\{(\mathbf{K}_m+\epsilon\mathbf{I})^{-1}\mathbf{K}_m\}_{m=1}^M$
requiring complexity in the order of $\mathcal{O}(N^3M)$. On the other hand, implementing GKMCCA when $N\ll D$ incurs complexity of order $\mathcal{O}(N^2MD)$, required to evaluate the $M$ kernel matrices.

\begin{remark}
	When the (non)linear maps $\bm{ \phi}_m(\cdot)$ needed to form the kernel matrices $\{\mathbf{K}_m\}_{m=1}^M$ in \eqref{eq:gmccak} are not given a priori, the multi-kernel methods are well motivated (see e.g., \cite{zhang2017going,2018scgmkl}). Concretely, one presumes that each $\mathbf{K}_m$ is a linear combination of $P$ kernel matrices, namely $\mathbf{K}_m=\sum_{p=1}^P\beta_m^p \mathbf{K}_{m}^p$, where $\{\mathbf{K}_m^p\}_{p=1}^P$ represent preselected view-specific kernel matrices for data $\mathbf{X}_m$. The unknown coefficients $\{\beta_m^p\ge 0\}_{m,p}$ are then jointly optimized with $\{\mathbf{A}_m\}_{m}$ and $\mathbf{S}$ in \eqref{eq:gmccak}. 
\end{remark}

\begin{remark}
	When more than one type of connectivity information on the common sources are available, our single graph-regularized MCCA schemes can be generalized to accommodate multiple or multi-layer graphs. Specifically, the single graph-based regularization term $\gamma{\rm Tr}(\mathbf{S}\mathbf{L}_{\mathcal{G}}\mathbf{S}^\top)$ in \eqref{eq:gmcca}, \eqref{eq:gmccadr}, and \eqref{eq:gmccak} can be replaced with $\sum_{i=1}^I\gamma_i {\rm Tr}(\mathbf{S}\mathbf{L}_{\mathcal{G}i}\mathbf{S}^\top)$ with possibly unknown yet learnable coefficients $\{\gamma_i \}_{i}$, where $\mathbf{L}_{\mathcal{G}i}$ denotes the graph Laplacian matrix of the $i$-th graph, for $i=1,\ldots,I$.
\end{remark}

\section{Numerical Tests}
In this section, numerical tests using real datasets are provided to showcase the merits of our proposed MCCA approaches in several machine learning applications, including user engagement prediction, friend recommendation, clustering, and classification.

\subsection{User engagement prediction}\label{sec:uep}
Given multi-view data of Twitter users, the goal of the so-called user engagement prediction is to determine which topics a Twitter user is likely to tweet about, by using hashtag as a proxy. The first experiment entails six datasets of Twitter users, which include EgoTweets, MentionTweets, FriendTweets, FollowersTweets, FriendNetwork, and FollowerNetwork data \footnote{Downloaded from http://www.dredze.com/datasets/multiviewembeddings/.},  where $\{D_m=1,000\}_{m=1}^6$ and $N=1,770$ users' data are randomly chosen from the database. Details in generating those multiview data can be found in \cite{tweeterdata}. Based on data $\{\mathbf{X}_m\in \mathbb{R}^{D_m\times N}\}_{m=1}^3$ from the first $3$ views, three adjacency matrices $\{\mathbf{W}_m\in\mathbb{R}^{N\times N}\}_{m=1}^3$ are constructed, whose $(i,\,j)$-th entries are 
\begin{equation}\label{eq:wij}
w_{ij}^m:=\left\{
\begin{array}{ll}
\mathbf{K}^t_m(i,\,j), & i\in\mathcal{N}_{k_1}(j) {~\rm or~}j\in\mathcal{N}_{k_1}(i)\\
0,& {\rm otherwise}
\end{array}
\right.
\end{equation}
 where $\mathbf{K}_m^t$ is a Gaussian kernel matrix of $\mathbf{X}_m$ with bandwidth equal to the mean of the corresponding Euclidean distances, and $\mathcal{N}_{k_1}(j)$ the set of column indices of $\mathbf{K}_m^t$ containing the $k_1$-nearest neighbors of column $j$. Our graph adjacency matrix is built using $\mathbf{W}=\sum_{m=1}^3\mathbf{W}_m$. 
 To perform graph (G) PCA \cite{gpca} and PCA, six different views of the data are concatenated to form a single dataset of $6,000$-dimensional data vectors. 

\begin{figure}[t]
	\centering 
	\includegraphics[scale=0.58]{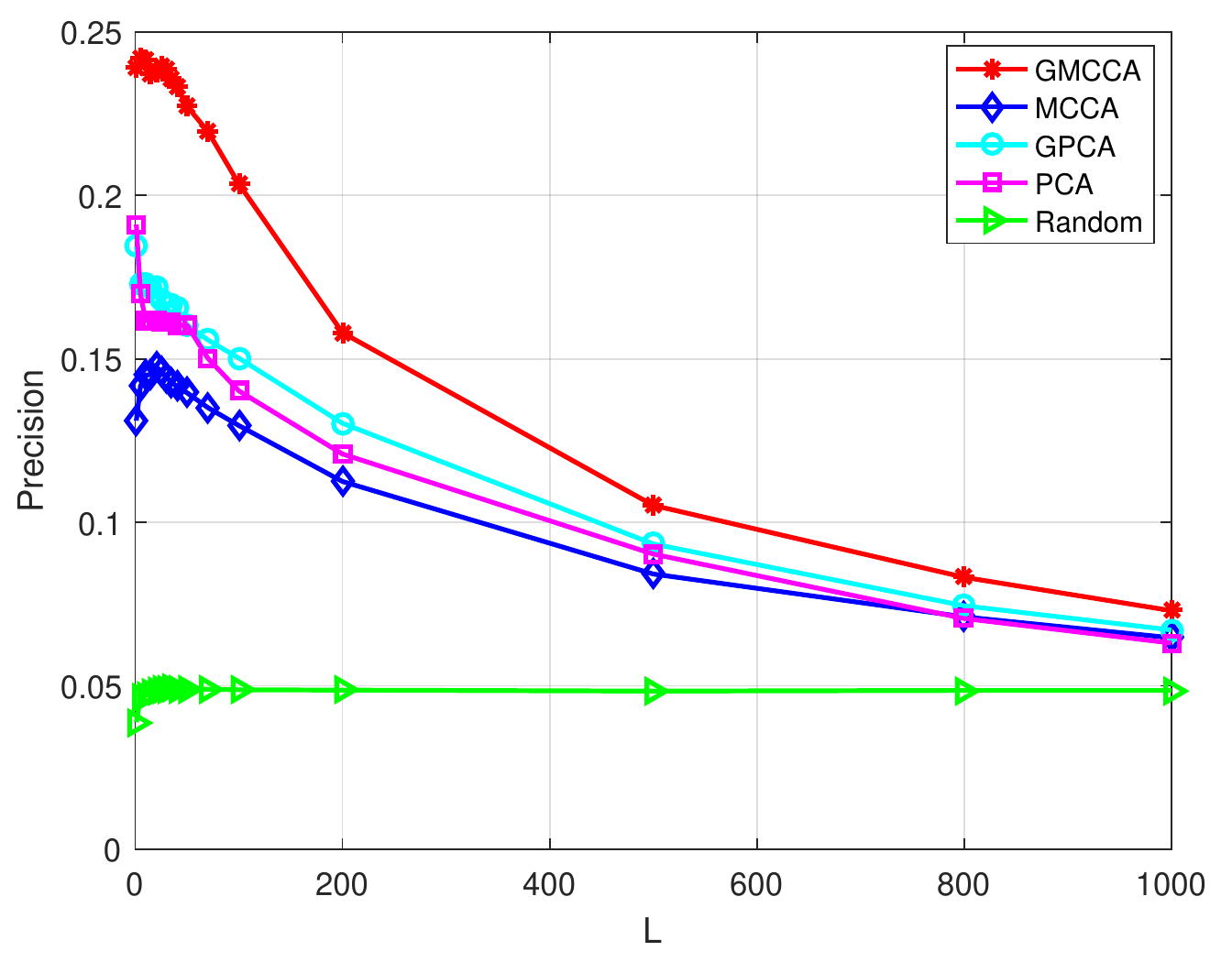}  
	\vspace{-5pt}
	\caption{\small{Precision of user engagement prediction.}}
	\label{fig:ht_precision}
	\vspace{-5pt}
\end{figure}

\begin{figure}[t]
	\centering 
	\includegraphics[scale=0.58]{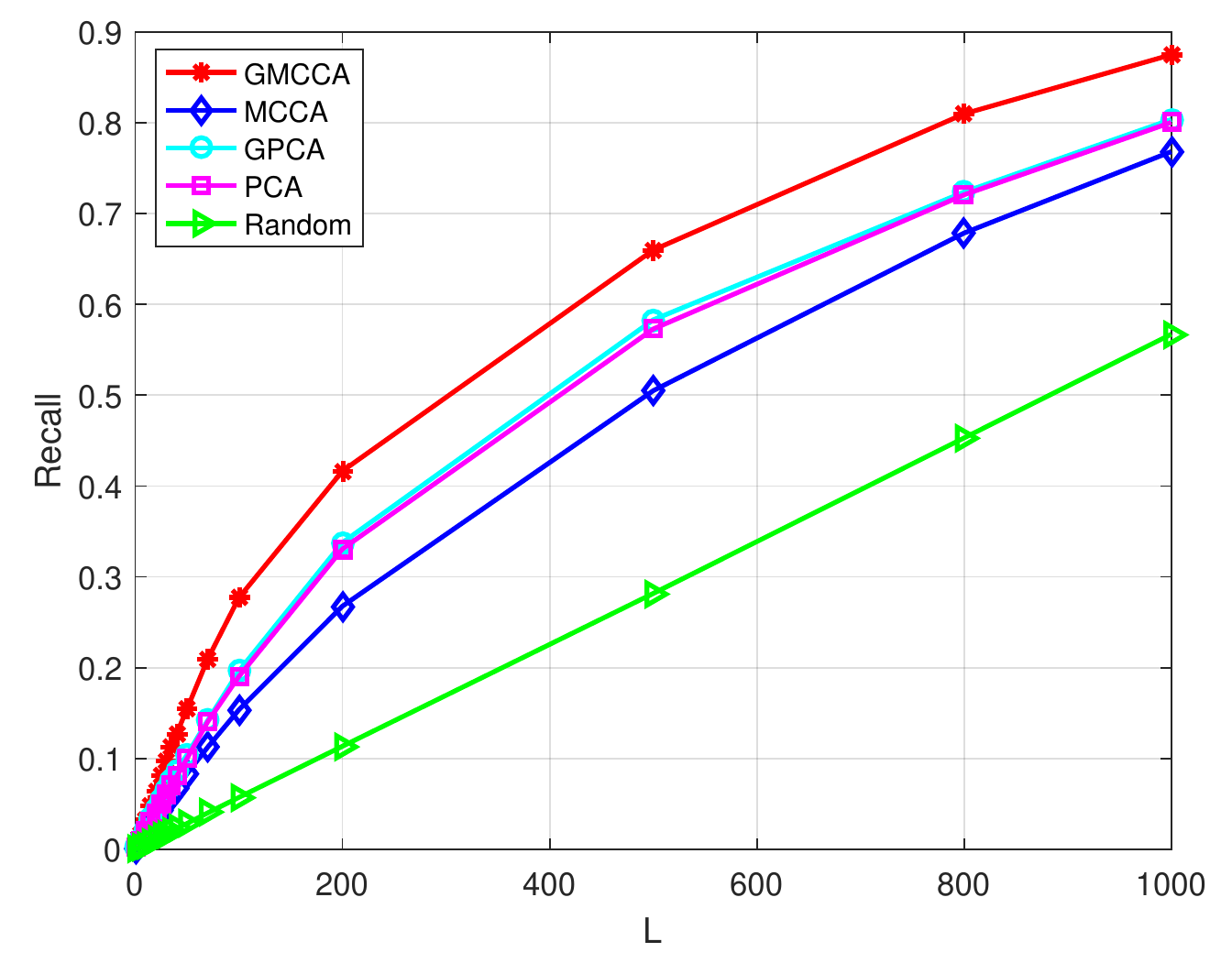}  
	\vspace{-5pt}
	\caption{\small{Recall of user engagement prediction.}}
	\label{fig:ht_recall}
	\vspace{-5pt}
\end{figure}

We selected $9$ most frequently used hashtags. Per Monte Carlo (MC) run, $5$ users who tweeted each selected hashtag were randomly chosen as exemplars of users that would employ this hashtag in the future. All other users that tweeted each hashtag were ranked by the cosine distance of their representations to the average representation of those $5$ users, where the representation per user is either the corresponding estimate of the common source obtained by (G)MCCA or the principal components by (G)PCA. Before computing cosine distance, the $d$-dimensional representations were z-score normalized. In other words,  each dimension has its mean removed, and subsequently scaled to have unit variance. The representations are learned on data collected pre-March 2015, while the association between hashtags and users is extracted in March 2015. This implies that the hashtags do not impact the representation learning. Pertinent hyper-parameters were set as $k_1=10$, $\gamma=0.05$, and $d=5$.

Prediction performance is evaluated using three metrics: precision, recall, and mean reciprocal rank (MRR), where a user is marked as correct if this user uses the hashtag. The precision is defined as the ratio of the number of correctly predicted users over the total number of predicted users considered. Recall is the ratio of the number of correctly predicted users over the total number of users that use the hashtag. MRR is the average inverse of the ranks of the first correctly predicted users. 

Figures \ref{fig:ht_precision} and \ref{fig:ht_recall} present the average precision and recall of GMCCA, MCCA, GPCA, PCA, and a random ranking scheme over $100$ MC realizations, with a varying number $L$ of evaluated users per hashtag. Here, the random ranking is included as a baseline. Table \ref{tab:ht} reports the prediction performance of simulated schemes with $L=35$ being fixed. Clearly, GMCCA outperforms its competing alternatives in this Tweeter user engagement prediction task. Moreover, ranking through all approaches is consistent across precision, recall, and MRR. 

\renewcommand{\arraystretch}{2} 
\begin{table}[tp]	
	\centering
	\fontsize{9.5}{9}\selectfont
	\caption{User engagement prediction performance.}
	\vspace{-3pt}
	\label{tab:ht}
	\vspace{.8em}
	\begin{tabular}{|c|c|c|c|c|}
		\hline
		Model&Precision&Recall&MRR\cr
		\hline
		\hline
		GMCCA&0.2357&0.1127&0.4163\cr\hline
		MCCA&0.1428&0.0593&0.2880\cr\hline
		GPCA&0.1664&0.0761&0.3481\cr\hline
		PCA&0.1614&0.0705&0.3481\cr
		\hline
		Random&0.0496&0.0202&0.1396\cr
		\hline
	\end{tabular}
\end{table}

\subsection{Friend recommendation}
GMCCA is further examined for friend recommendation, where the graph can be constructed from an alternative view of the data, as we argued in Remark \ref{rmk:graphcca}. Specifically for this test, $3$ Tweeter user datasets \cite{tweeterdata} from $2,506$ users were used to form $\{\mathbf{X}_m\in\mathbb{R}^{1,000}\}_{m=1}^3$, which are EgoTweets, FollowersTweets, and FollowerNetwork data. An alternative view, the FriendTweets data of the same group of users, was used to construct the common source graph. The weight matrix $\mathbf{W}$ is obtained following a similar way to form $\mathbf{W}_m$ but replacing $\mathbf{K}_m^t$ in \eqref{eq:wij} with a Gaussian kernel matrix of FriendTweets data.

In the experiment, $20$ most popular accounts were selected, which correspond to celebrities. Per realization, $10$ users who follow each celebrity were randomly picked, and all other users were ranked by their cosine distances to the average of the $10$ picked representations. We z-score normalize all representations before calculating the cosine distances. The same set of evaluation criteria as in user engagement prediction in Sec. \ref{sec:uep} was adopted here, where a user is considered to be a correctly recommended friend if both follow the given celebrity. Hyper-parameters $k_1=50$, $\gamma=0.05$, and $d=5$ were simulated. 
The friend recommendation performance of GMCCA, MCCA, GPCA, PCA, and Random ranking is evaluated after averaging over $100$ independent realizations.

In Figs. \ref{fig:fr_precision} and \ref{fig:fr_recall}, the precision and recall of all simulated algorithms under an increasing number of recommended friends ($L$) are reported. Plots corroborate the advantages of our GMCCA relative to its simulated alternatives under different numbers of recommendations. Moreover, Table \ref{tab:fr} compares the precision, recall, and MRR of simulated schemes for fixed $L=100$. Regarding the results, we have the following observations: i) GMCCA is more attractive in the recommendation task than its alternatives; ii) precision and recall differences among approaches are consistent for different $L$ values; and, iii) ranking achieved by these schemes is consistent across $3$ metrics for fixed $L=100$.

\begin{figure}[t]
	\centering 
	\includegraphics[scale=0.58]{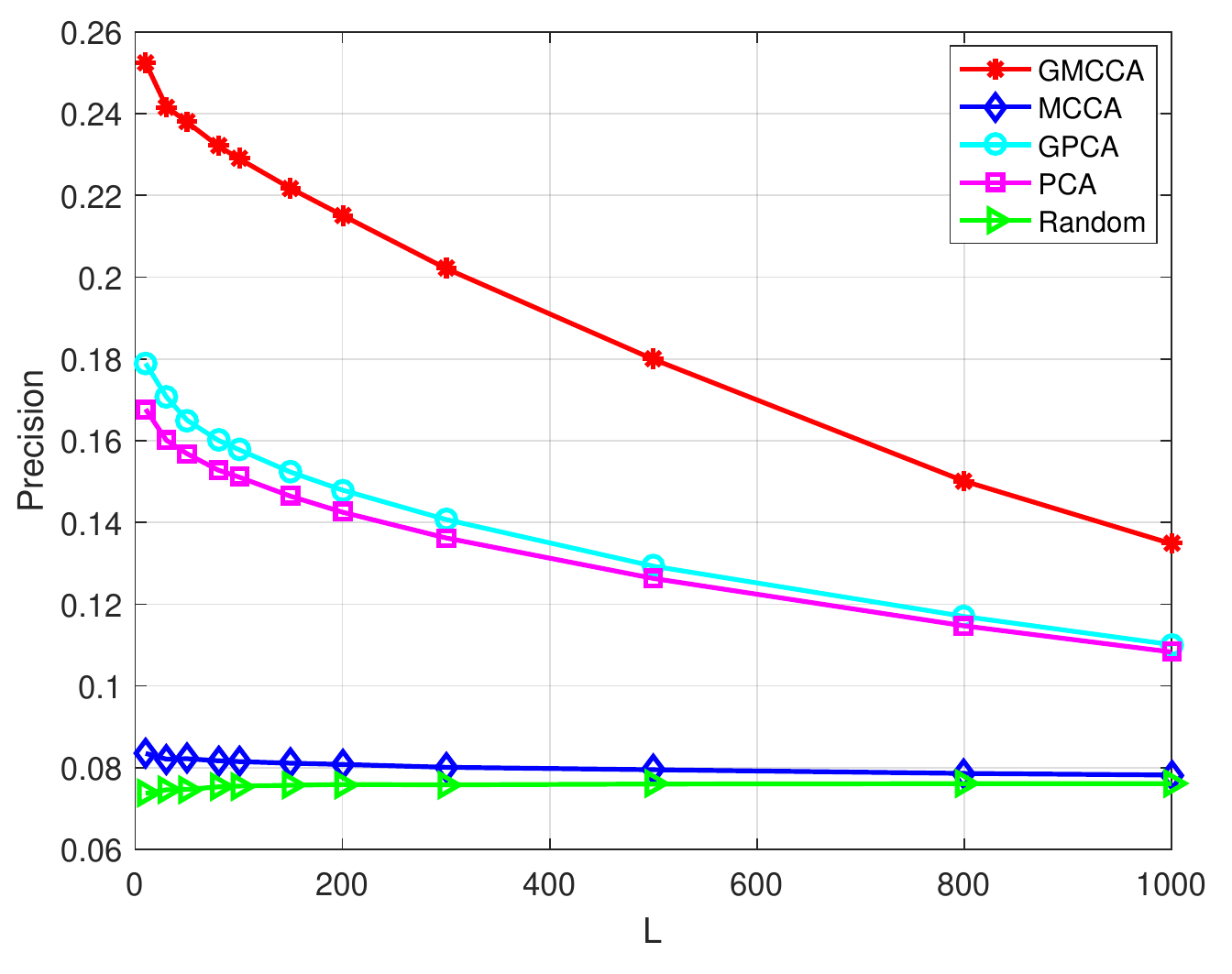}  
	\vspace{-6pt}
	\caption{\small{Precision of friend recommendation.}}
	\label{fig:fr_precision}
	\vspace{-5pt}
\end{figure}

\begin{figure}[t]
	\centering 
	\includegraphics[scale=0.58]{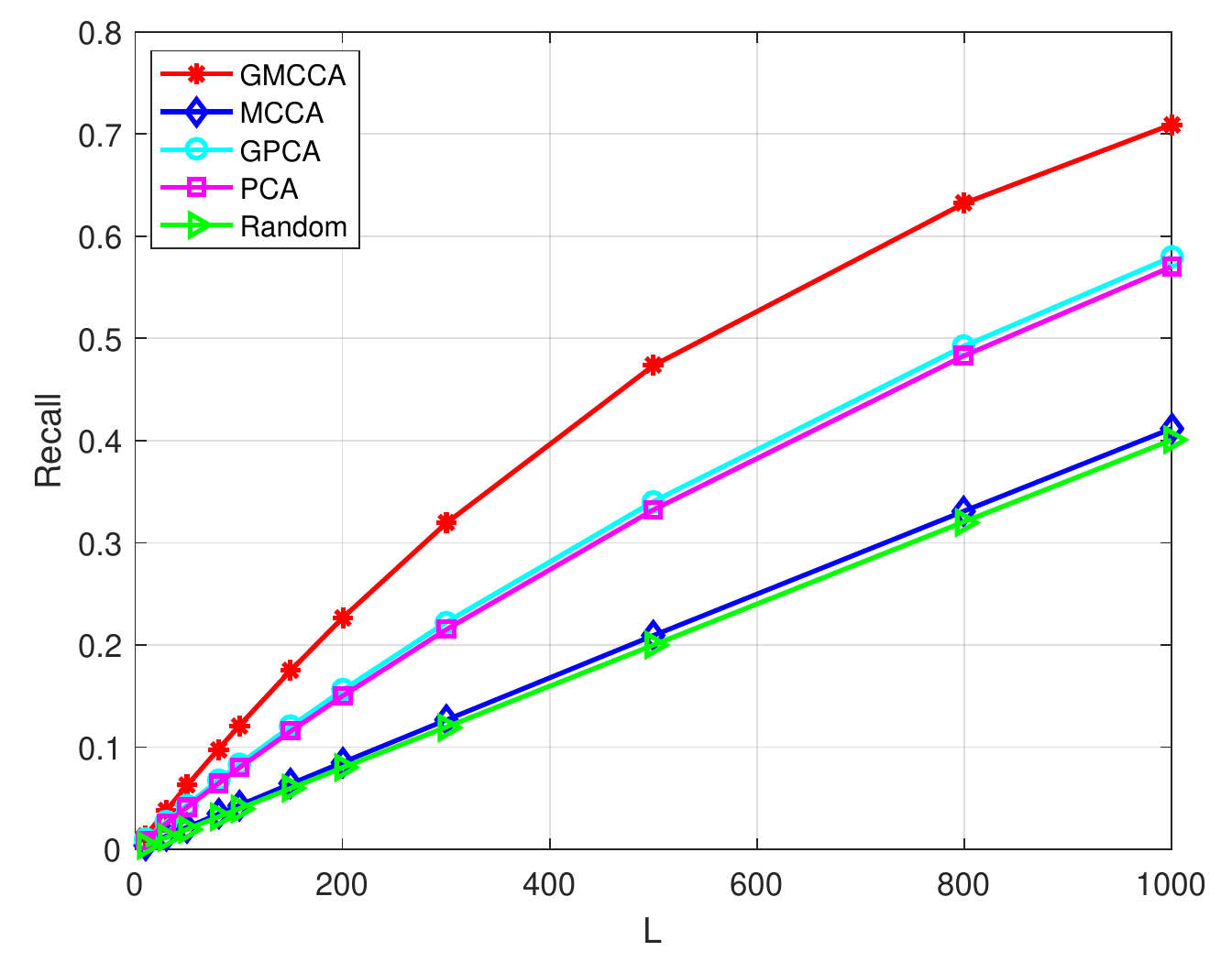}  
	\vspace{-6pt}
	\caption{\small{Recall of friend recommendation.}}
	\label{fig:fr_recall}
	\vspace{-5pt}
\end{figure}

\renewcommand{\arraystretch}{2} 
\begin{table}[tp]	
	\centering
	\fontsize{9.5}{9}\selectfont
	\caption{Friend recommendation performance comparison.}
	\label{tab:fr}
	\vspace{.8em}
	\begin{tabular}{|c|c|c|c|c|}
		\hline
		Model&Precision&Recall&MRR\cr
		\hline
		\hline
		GMCCA&0.2290&0.1206&0.4471\cr\hline
		MCCA&0.0815&0.0429&0.2225\cr\hline
		GPCA&0.1578&0.0831&0.3649\cr\hline
		PCA&0.1511&0.0795&0.3450\cr
		\hline
		Random&0.0755&0.0397&0.2100\cr
		\hline
	\end{tabular}
	\vspace{-10pt}
\end{table}

\begin{figure*}[th!]
	\centering
	\hspace{0.1in}
	\subfigure[GMCCA]{
		\label{fig:subfig:a} 
		\includegraphics[width=4.22cm,height=3.84cm]{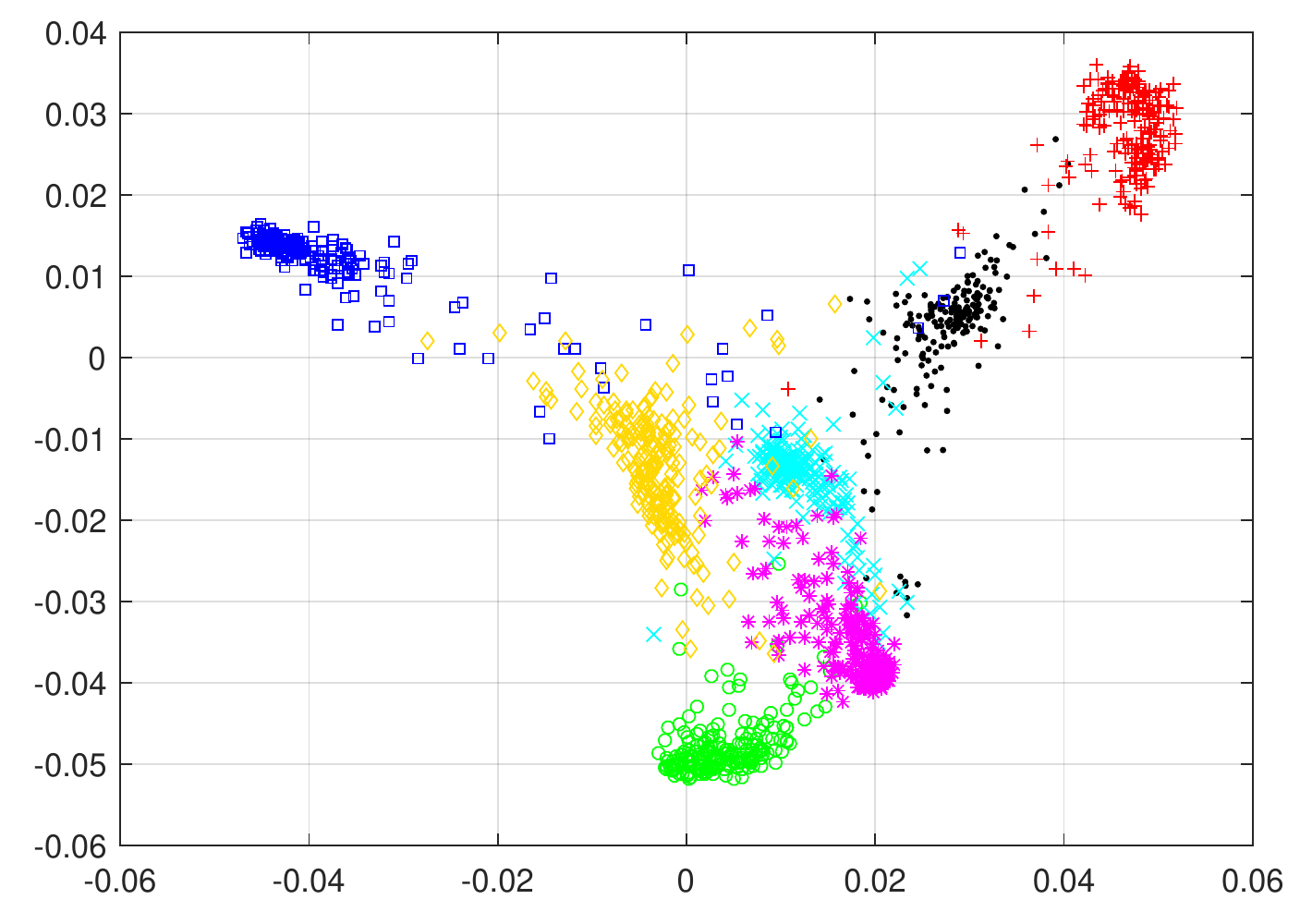}}
	\subfigure[MCCA]{
		\label{fig:subfig:b} 
		\includegraphics[width=4.22cm,height=3.84cm]{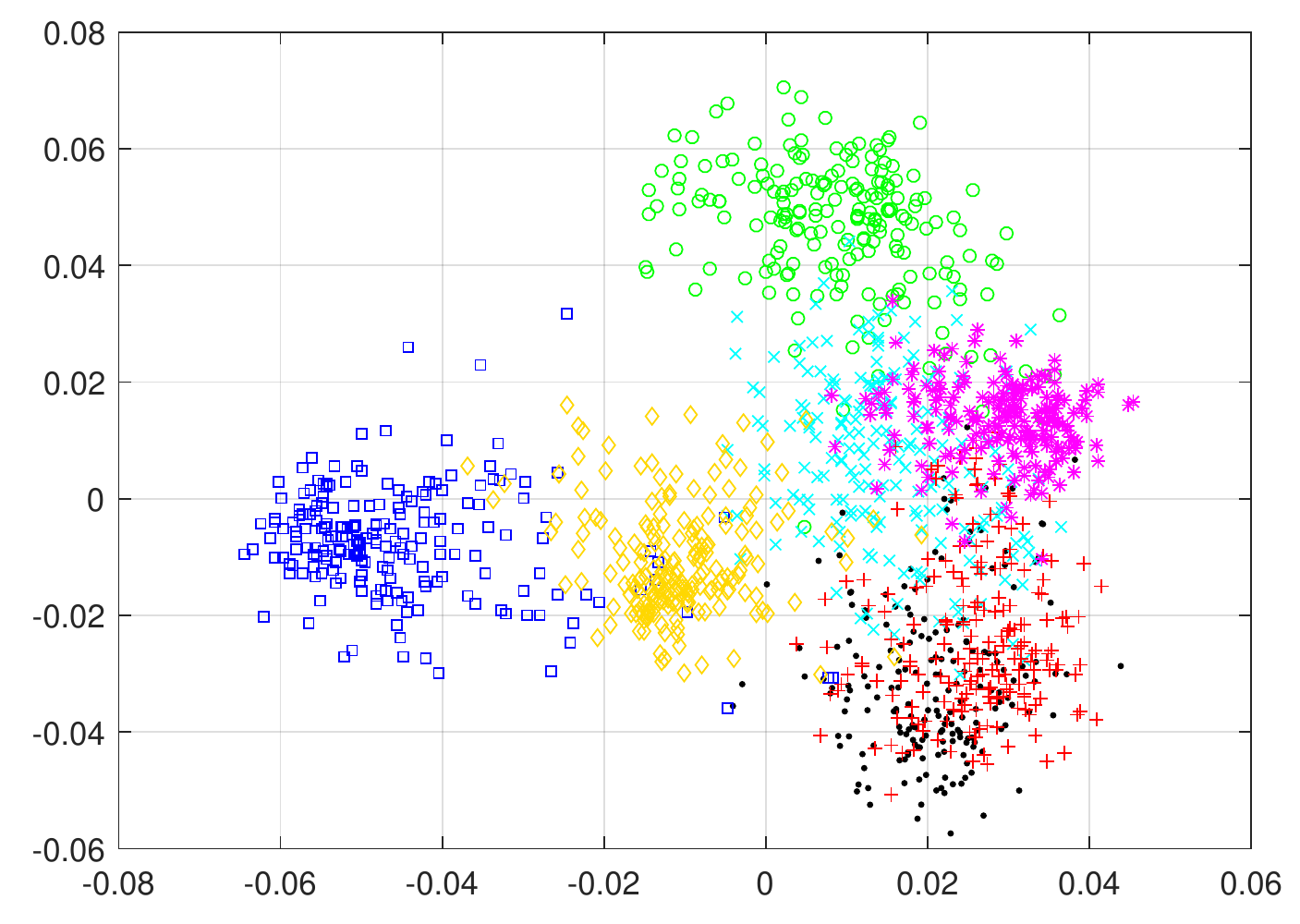}}
	\subfigure[GPCA]{
		\label{fig:subfig:c} 
		\includegraphics[width=4.22cm,height=3.84cm]{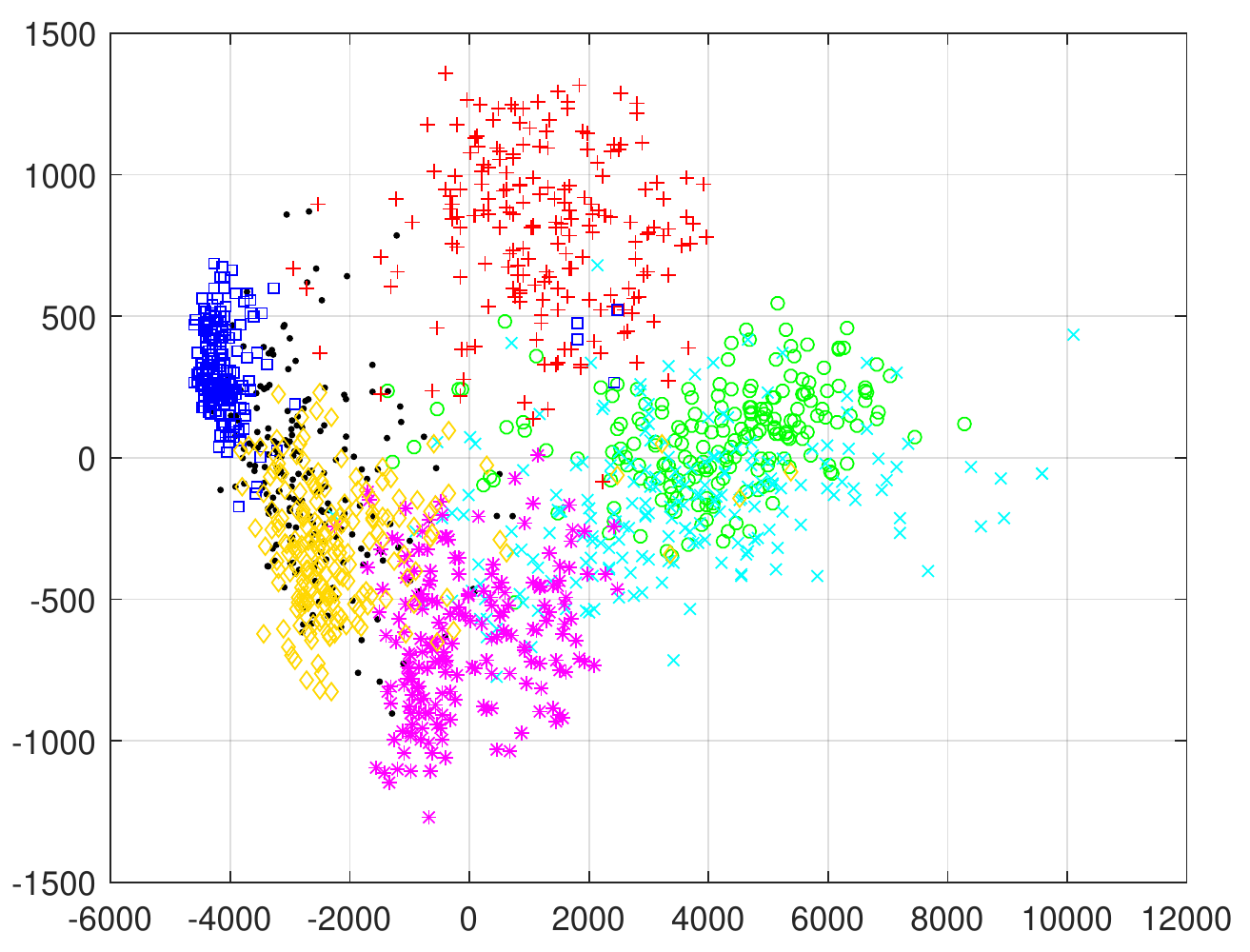}}
	\subfigure[PCA]{
		\label{fig:subfig:c} 
		\includegraphics[width=4.22cm,height=3.84cm]{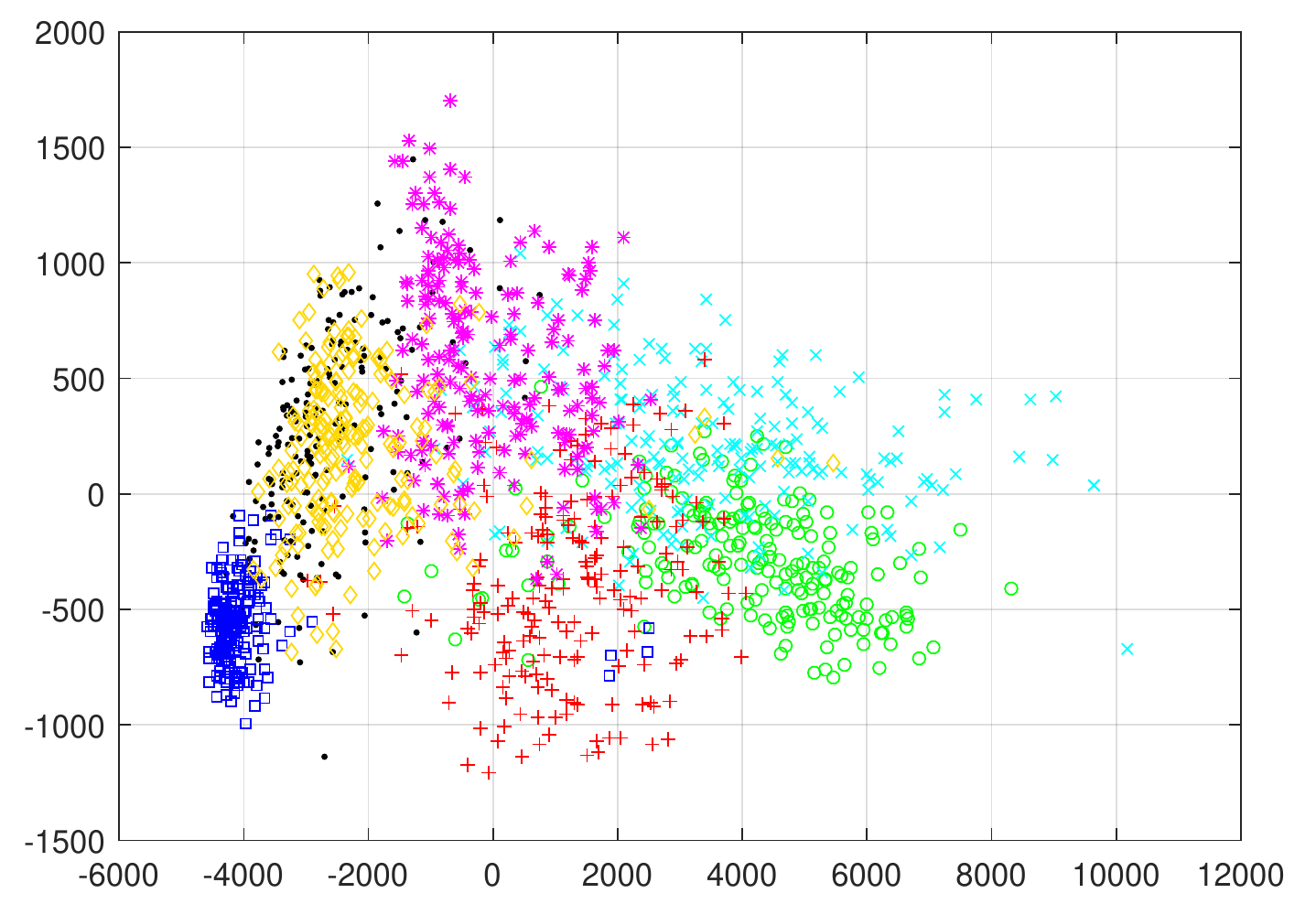}}
	\caption{Scatter plot of the first two rows of $\hat{\mathbf{S}}$ or principal components.}
	\label{fig:subfig} 
\end{figure*}

\subsection{UCI data Clustering}\label{simu:uci}
Handwritten digit data from the UCI machine learning repository\footnote{http://archive.ics.uci.edu/ml/datasets/Multiple+Features.}
were called for to assess GMCCA for clustering. This dataset contains $6$ feature sets of $10$ classes corresponding to $10$ digits from $0$ to $9$, as listed in Table \ref{tab:uci}. There are $200$ data per class ($2,000$ in total) per feature set. 
Seven clusters of data including digits $1,\,2,\,3,\,4,\,7,\,8$, and $9$ were used to form the views $\{\mathbf{X}_m\in\mathbb{R}^{D_m\times 1,400}\}_{m=1}^6$ with $D_1=76$, $D_2=216$, $D_3=64$, $D_4=240$, $D_5=47$, and $D_6=6$. The graph adjacency matrix is constructed using \eqref{eq:wij}, after substituting $\mathbf{K}_m^t$ by the Gaussian kernel matrix of $\mathbf{X}_3$. GPCA and PCA were performed on the concatenated data vectors of dimension $\sum_{m=1}^6D_m$, while the 
 K-means was performed using either $\hat{\mathbf{S}}$, or the principal components with $\gamma=0.1$ and $d=3$. 

Clustering performance is evaluated in terms of two metrics, namely clustering accuracy and scatter ratio. Clustering accuracy is the percentage of correctly clustered samples. Scatter ratio is defined as $C_t/\sum_{i=1}^7 C_i$, where $C_t$ and $C_i$ denote the total scatter value and the within-cluster scatter value, given correspondingly by $C_t:=\|\hat{\mathbf{S}}\|_F^2$ and $C_i:=\sum_{j\in\mathcal{C}_i}\|\hat{\mathbf{s}}_j-\frac{1}{|\mathcal{C}_i|}\sum_{\ell\in\mathcal{C}_i}\hat{\mathbf{s}}_{\ell}\|_2^2$; here, $\mathcal{C}_i$ is the set of data vectors belonging to the $i$-th cluster, and $|\mathcal{C}_i|$ is the cardinality of $\mathcal{C}_i$. 

Table \ref{tab:cluster} reports the clustering performance of MCCA, PCA, GMCCA, and GPCA for different $k_1$ values. Clearly, GMCCA yields the highest clustering accuracy and scatter ratio. Fixing $k_1=50$, Fig. \ref{fig:subfig} plots the first two dimensions of the common source estimates obtained by (G)MCCA along with the first two principal components of (G)PCA, with different colors signifying different clusters. As observed from the scatter plots, GMCCA separates the $7$ clusters the best, in the sense that data points within clusters are concentrated but across clusters are far apart. 

\begin{table}
	\centering
	\fontsize{9.5}{8}\selectfont
	\caption{Six sets of features of handwritten numerals.}
	\label{tab:uci}
	\begin{tabular}{ll}
		\hline
		mfeat-fou&76-dim. Fourier coeff. of character shapes features\\ 
		mfeat-fac&216-dim. profile correlations features\\
		mfeat-kar&64-dim. Karhunen-Love coefficients features\\
		mfeat-pix&240-dim. pixel averages in 2 x 3 windows\\
		mfeat-zer&47-dim. Zernike moments features\\
		mfeat-mor&6-dim. morphological features\\
		\hline
	\end{tabular}
\end{table}

\renewcommand{\arraystretch}{2} 
\begin{table}[tp]	
	\centering
	\fontsize{9.5}{9}\selectfont
	\caption{Clustering performance comparison.}
	\label{tab:cluster}
	\vspace{.5em}
	\begin{tabular}{|c|c|c|c|c|}
		\hline
		\multirow{2}{*}{$k_1$}&
		\multicolumn{2}{c|}{Clustering accuracy}&\multicolumn{2}{c|}{ Scatter ratio}\cr\cline{2-5}
		&GMCCA&GPCA&GMCCA&GPCA\cr
		\hline
		\hline
		10&0.8141&0.5407&9.37148&4.9569\cr\hline
		20&0.8207&0.5405&11.6099&4.9693\cr\hline
		30&0.8359&0.5438&12.2327&4.9868\cr\hline
		40&0.8523&0.5453&12.0851&5.0157\cr\hline
		50&0.8725&0.5444&12.1200&5.0640\cr
		\hline
		\hline
		\multirow{1}{*}{MCCA}&
		\multicolumn{2}{c|}{0.8007}&\multicolumn{2}{c|}{ 5.5145}\cr\cline{1-2}
		\hline
		\multirow{1}{*}{PCA}&
		\multicolumn{2}{c|}{0.5421}&\multicolumn{2}{c|}{ 4.9495}\cr\cline{1-2}
		\hline
	\end{tabular}
\end{table}

\begin{figure}[h]
	\centering 
	\includegraphics[scale=0.58]{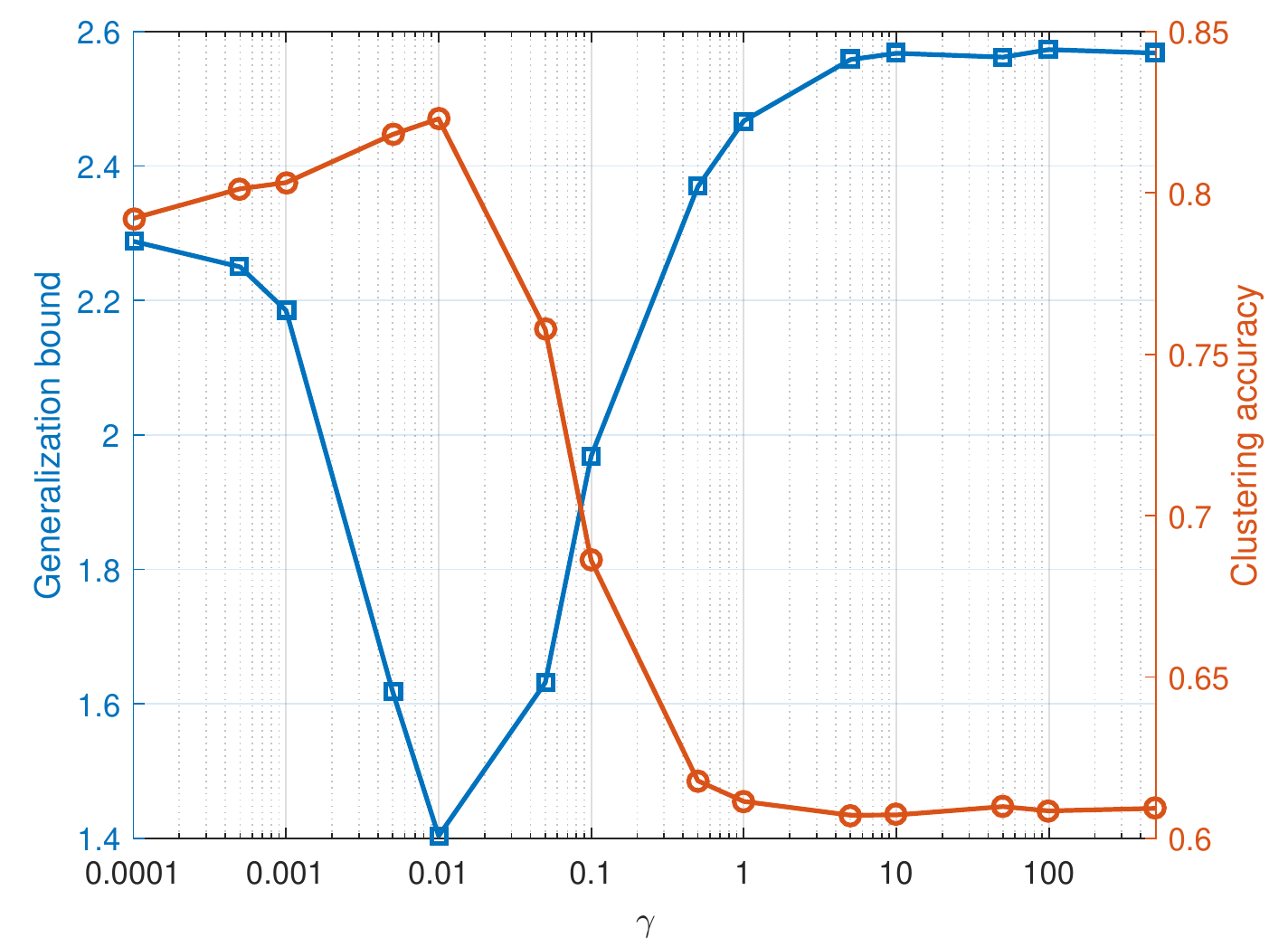}  
\vspace{-6pt}
	\caption{\small{Generalization bound versus $\gamma$.}}
	\label{fig:gb}
	\vspace{-5pt}
\end{figure}

\subsection{Generalization bound versus $\gamma$}
Here, we wish to demonstrate the usefulness of the generalization bound of GMCCA derived in Sec. \ref{sec:gb}. Specifically, we will test numerically the effect of $\gamma$ on the generalization error bound defined on the right hand side (RHS) of \eqref{eq:thm}.

In this experiment, $20$ MC simulations were performed to evaluate the clustering performance of GMCCA using the UCI dataset described in Sec. \ref{simu:uci}. Per MC realization, $200$ samples per cluster were randomly and evenly divided to obtain training data $\{\mathbf{X}_m^{\rm tr}\in\mathbb{R}^{D_m\times 700}\}_{m=1}^3$ 
and testing data $\{\mathbf{X}_m^{\rm te}\in\mathbb{R}^{D_m\times 700}\}_{m=1}^3$. 
The same $7$ digits in Sec. \ref{simu:uci} and their first $3$ views were employed. GMCCA was performed on the training data to obtain $\{\hat{\mathbf{U}}_m\in\mathbb{R}^{D_m\times 3}\}_{m=1}^3$. Subsequently, low-dimensional representations of the testing data were found as $\sum_{m=1}^3\hat{\mathbf{U}}_m^\top \mathbf{X}_m^{\rm te} \in\mathbb{R}^{3\times 700}$, which were fed into the K-means for digit clustering. The generalization bound was evaluated utilizing the RHS of  \eqref{eq:thm}, where $\delta=0.1$, and $B=\sqrt{\sum_{m=1}^2\sum_{m'=m+1}^3 \|\hat{\mathbf{U}}_m^\top\hat{\mathbf{U}}_m+\hat{\mathbf{U}}_{m'}^\top\hat{\mathbf{U}}_{m'}\|_F^2}$.

Figure \ref{fig:gb} depicts the average generalization error bound along with clustering accuracy on the test data for different $\gamma$ values ranging from $0$ to $500$. Interestingly, at $\gamma=0.01$, the bound attains its minimum, and at the same time, the clustering accuracy achieves its maximum. This indeed provides us with an effective way to select the hyper-parameter value for our GMCCA approaches.


\subsection{Face recognition}\label{simu:face}
The ability of GDMCCA in face recognition is evaluated using the Extended Yale-B (EYB) face image database \cite{yaleb}. The EYB database contains frontal face images of $38$ individuals, each having $65$ images of $192\times 168$ pixels. Per MC realization, we performed Coiflets, Symlets, and Daubechies orthonormal wavelet transforms on $20$ randomly selected individuals' images to form three feature datasets. Subsequently, three feature matrices of each image were further resized to $50\times 40$ pixels, followed by vectorization to obtain three $2,000\times 1$ vectors. For each individual, $N_{\rm tr}$ images were randomly chosen, and the corresponding three sets of wavelet transformed data were used to form the training datasets $\{\mathbf{X}_m\in\mathbb{R}^{2,000\times20N_{\rm tr} }\}_{m=1}^{3}$. Among the remaining images, $(30-0.5N_{\rm tr})$ images per individual were obtained to form the tuning datasets $\{\mathbf{X}^{\rm tu}_m\in\mathbb{R}^{2,000\times 20(30-0.5N_{\rm tr})}\}_{m=1}^3$, and another $(30-0.5N_{\rm tr})$ for testing $\{\mathbf{X}^{\rm te}_m\in\mathbb{R}^{2,000\times 20(30-0.5N_{\rm tr})}\}_{m=1}^3$, following a similar process to construct $\{\mathbf{X}_m\}_{m=1}^3$. 

The $20N_{\rm tr}$ original training images were resized to $50\times 40$ pixels, and subsequently vectorized to obtain $2,000\times 1$ vectors, collected as columns of $\mathbf{O}\in\mathbb{R}^{2,000\times 20N_{\rm tr}}$, which were further used to build $\mathbf{W}\in\mathbb{R}^{20N_{\rm tr}\times 20N_{\rm tr}}$. Per $(i,\,j)$-th entry of $\mathbf{W}$ is
\begin{equation}\label{eq:wijdual}
w_{ij}:=\left\{\begin{array}{cl}
\frac{\mathbf{o}_i^\top\mathbf{o}_j}{\|\mathbf{o}_i\|_2\|\mathbf{o}_j\|_2}, & i\in\mathcal{M}_{k_2}(j) {~\rm or~}j\in\mathcal{M}_{k_2}(i)\\
0,& {\rm otherwise}
\end{array}
\right.
\end{equation}
where $\mathbf{o}_i$ is the $i$-th column of $\mathbf{O}$, and $\mathcal{M}_{k_2}(i)$ the set of the $k_2$ nearest neighbors of $\mathbf{o}_i$ belonging to the same individual. 

In this experiment, $k_2=N_{\rm tr}-1$ was kept fixed.  
Furthermore, the three associated graph adjacency matrices in Laplacian regularized multi-view (LM) CCA \cite{blaschko2011semi} were built in a similar way to construct $\mathbf{W}$, after substituting $\mathbf{O}$ by $\{\mathbf{X}_m\}_{m=1}^3$ accordingly. The hyper-parameters in GDMCCA, DMCCA, GDPCA, LMCCA were tuned among $30$ logarithmically spaced values between $10^{-3}$ and $10^3$ to maximize the recognition accuracy on $\{\mathbf{X}^{\rm tu}_m\}_{m=1}^3$. 
After simulating GDMCCA, DMCCA, GDPCA, DPCA, and LMCCA, $10$ projection vectors were employed to find the low-dimensional representations of $\{\mathbf{X}_m^{\rm te}\}_{m=1}^3$. Subsequently, the $1$-nearest neighbor rule was applied for face recognition.

Figures \ref{fig:yaleview1}, \ref{fig:yaleview2}, and \ref{fig:yaleview3} describe the average recognition accuracies of GMDCCA, MDCCA, GDPCA, DPCA, LMCCA, and KNN, for testing data $\mathbf{X}_1^{\rm te}$, $\mathbf{X}_2^{\rm te}$, and $\mathbf{X}_3^{\rm te}$, respectively, and for a varying number $N_{\rm tr}$ of training samples over $30$ MC realizations. It is clear that the recognition performance of all tested schemes improves as $N_{\rm tr}$ grows. Moreover, GDMCCA yields the highest recognition accuracy in all simulated settings.

\begin{figure*}[th!]
	\centering
	\hspace{0.1in}
	\subfigure[1st view]{
		\label{fig:yaleview1} 
		\includegraphics[width=5.72cm,height=4.84cm]{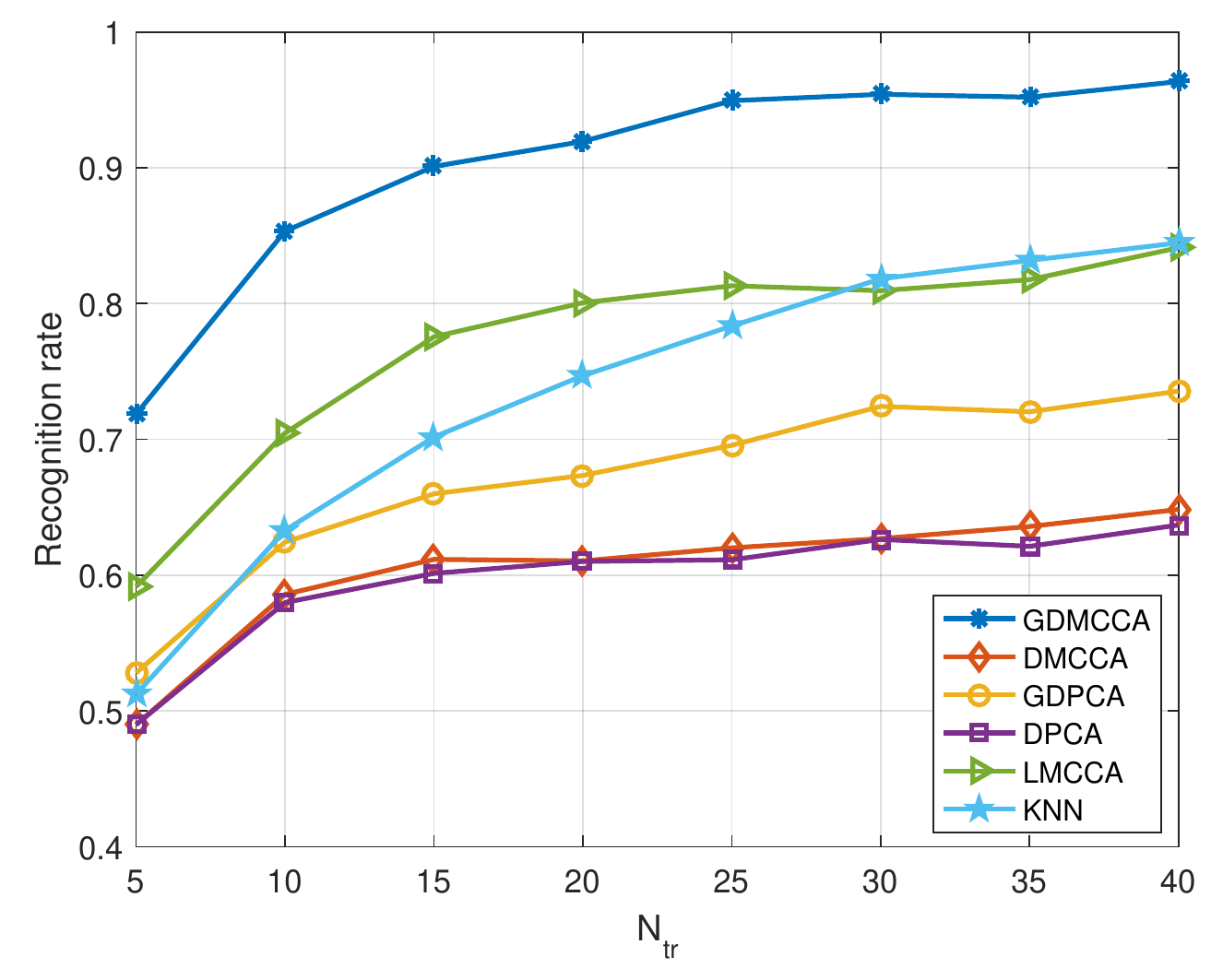}}
	\subfigure[2nd view]{
		\label{fig:yaleview2} 
		\includegraphics[width=5.72cm,height=4.84cm]{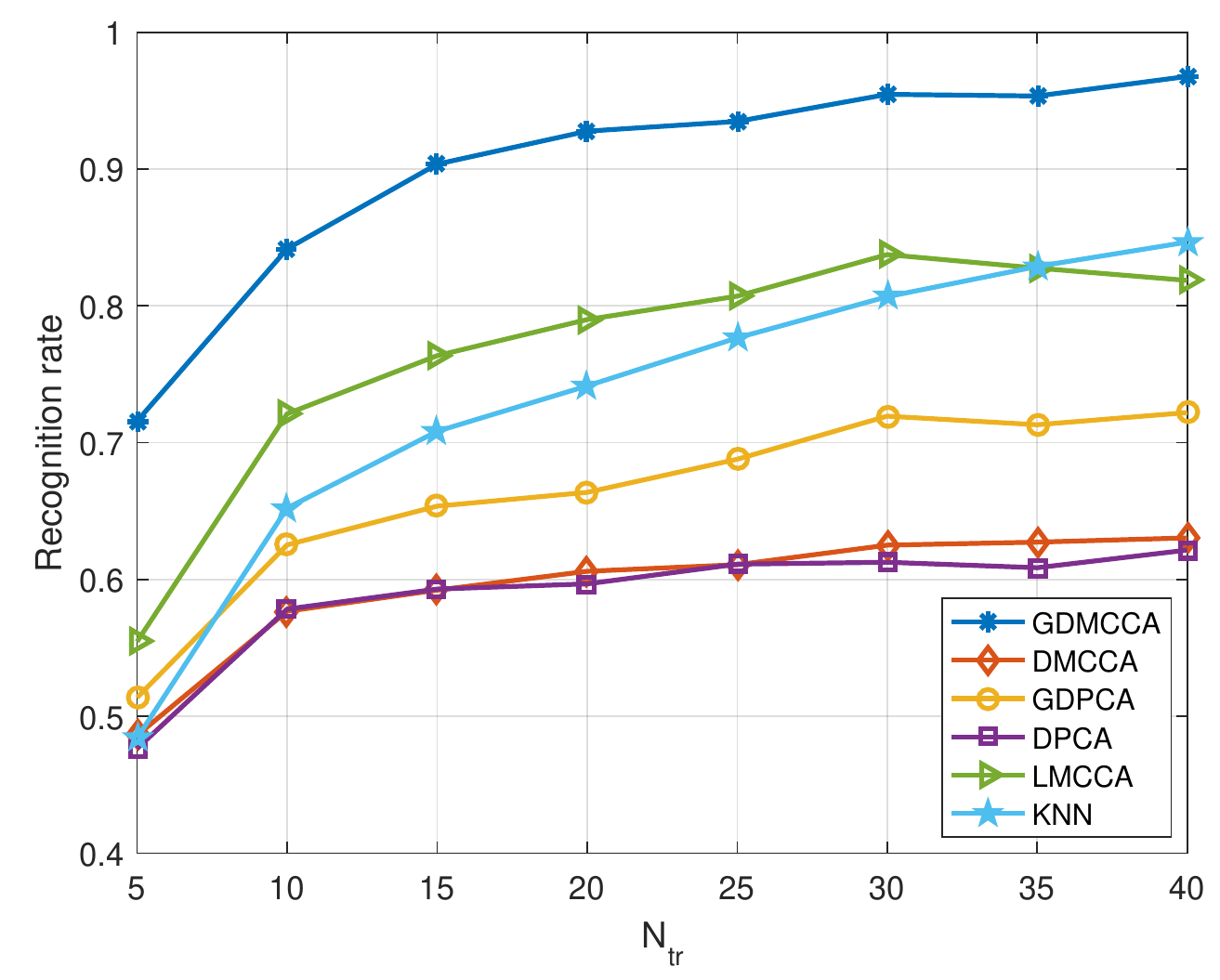}}
	\subfigure[3rd view]{
		\label{fig:yaleview3} 
		\includegraphics[width=5.72cm,height=4.84cm]{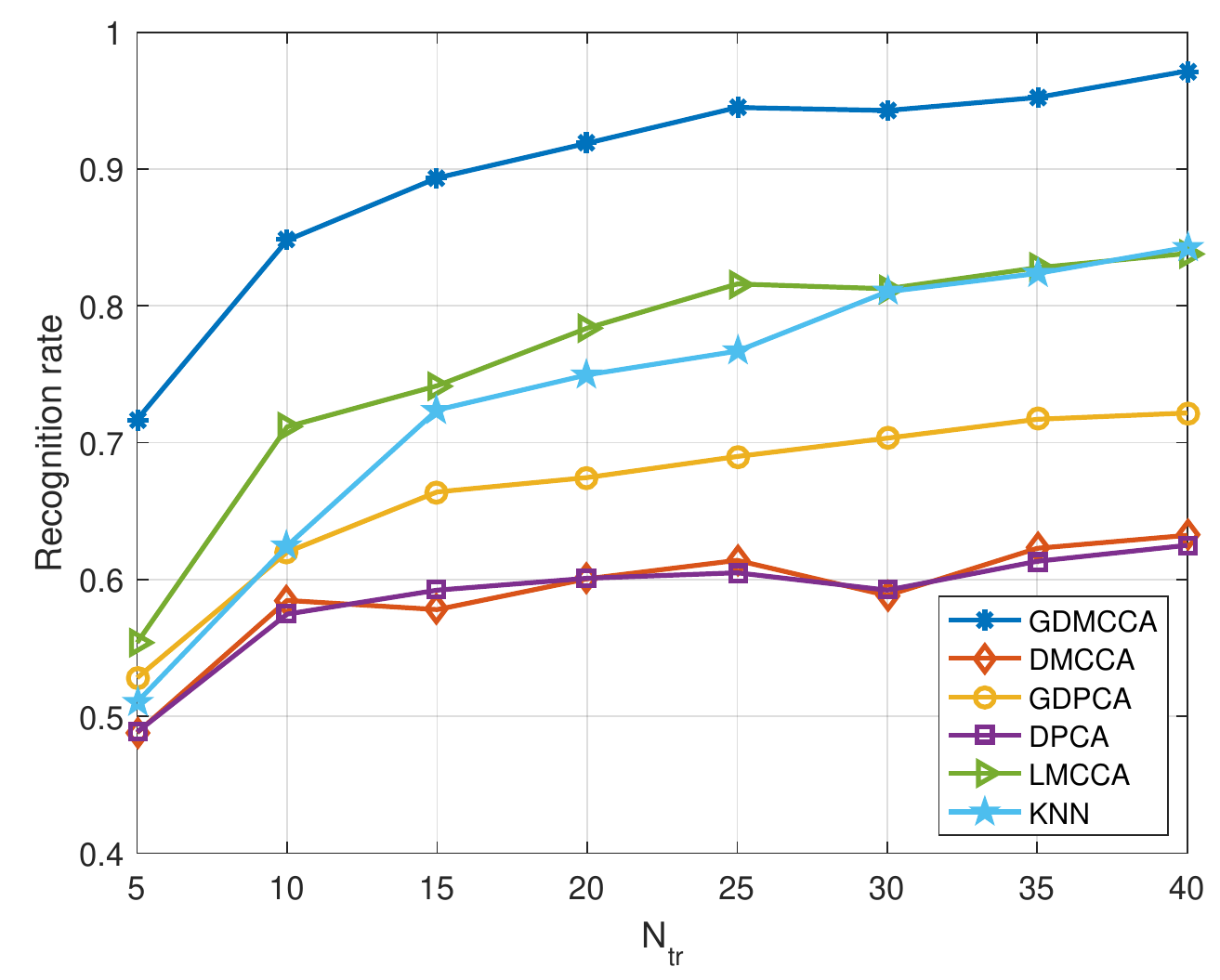}}
	\caption{Classification performance using   YEB data.}
	\label{fig:yale} 
\end{figure*}

\begin{figure*}[th!]
	\centering
	\hspace{0.1in}
	\subfigure[1st view]{
		\label{fig:mnist_view1} 
		\includegraphics[width=5.72cm,height=4.84cm]{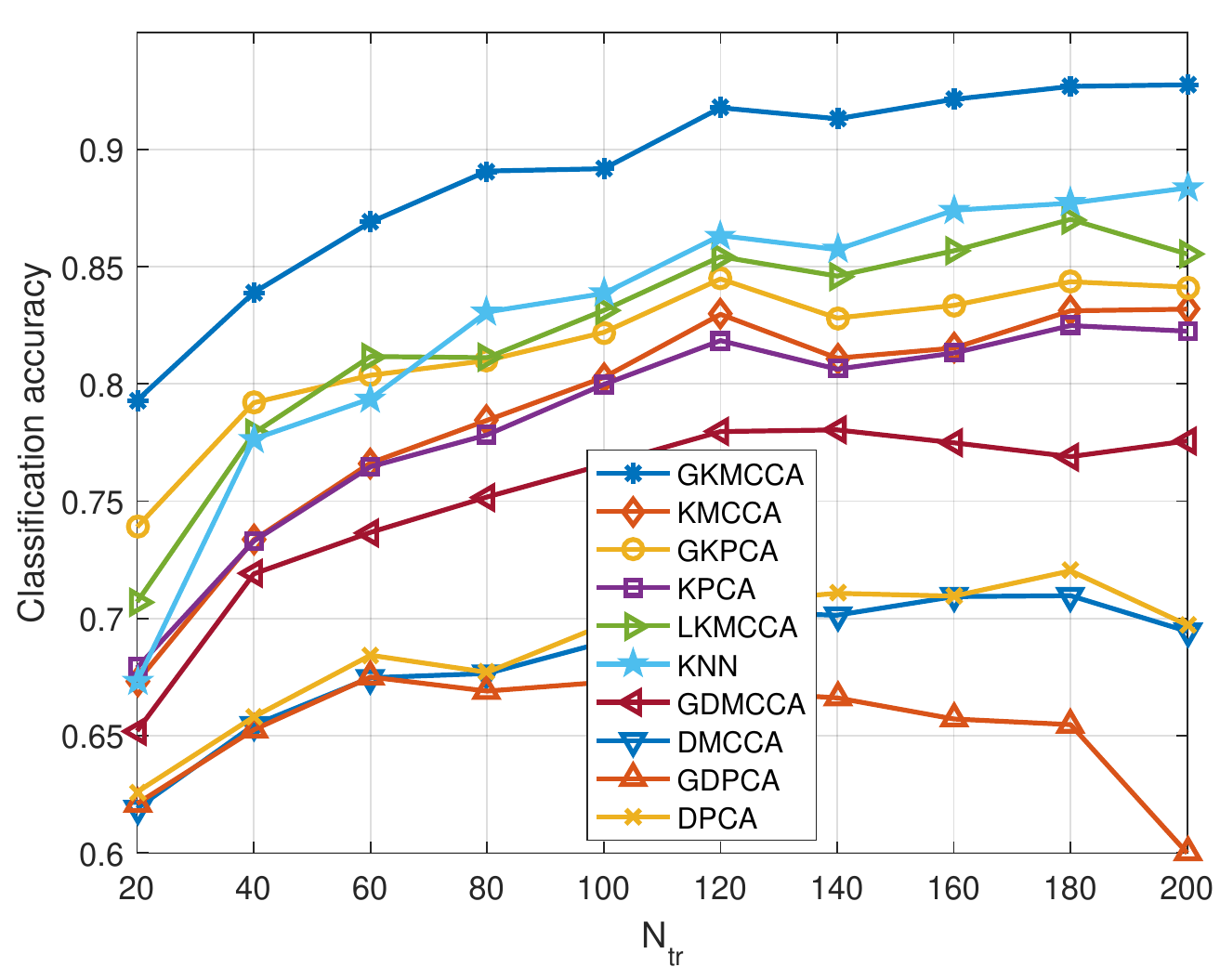}}
	\subfigure[2nd view]{
		\label{fig:mnist_view2} 
		\includegraphics[width=5.72cm,height=4.84cm]{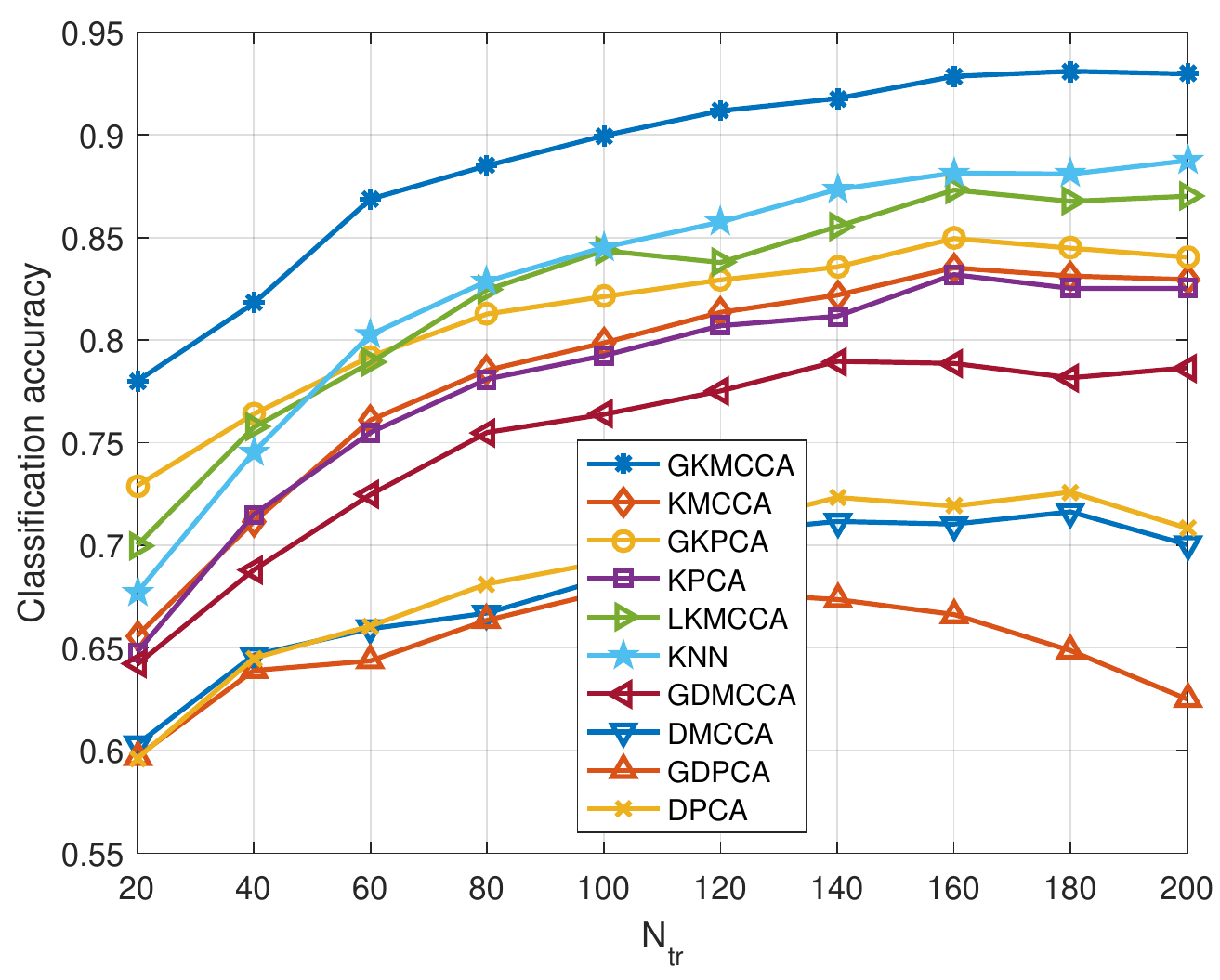}}
	\subfigure[3rd view]{
		\label{fig:mnist_view3} 
		\includegraphics[width=5.72cm,height=4.84cm]{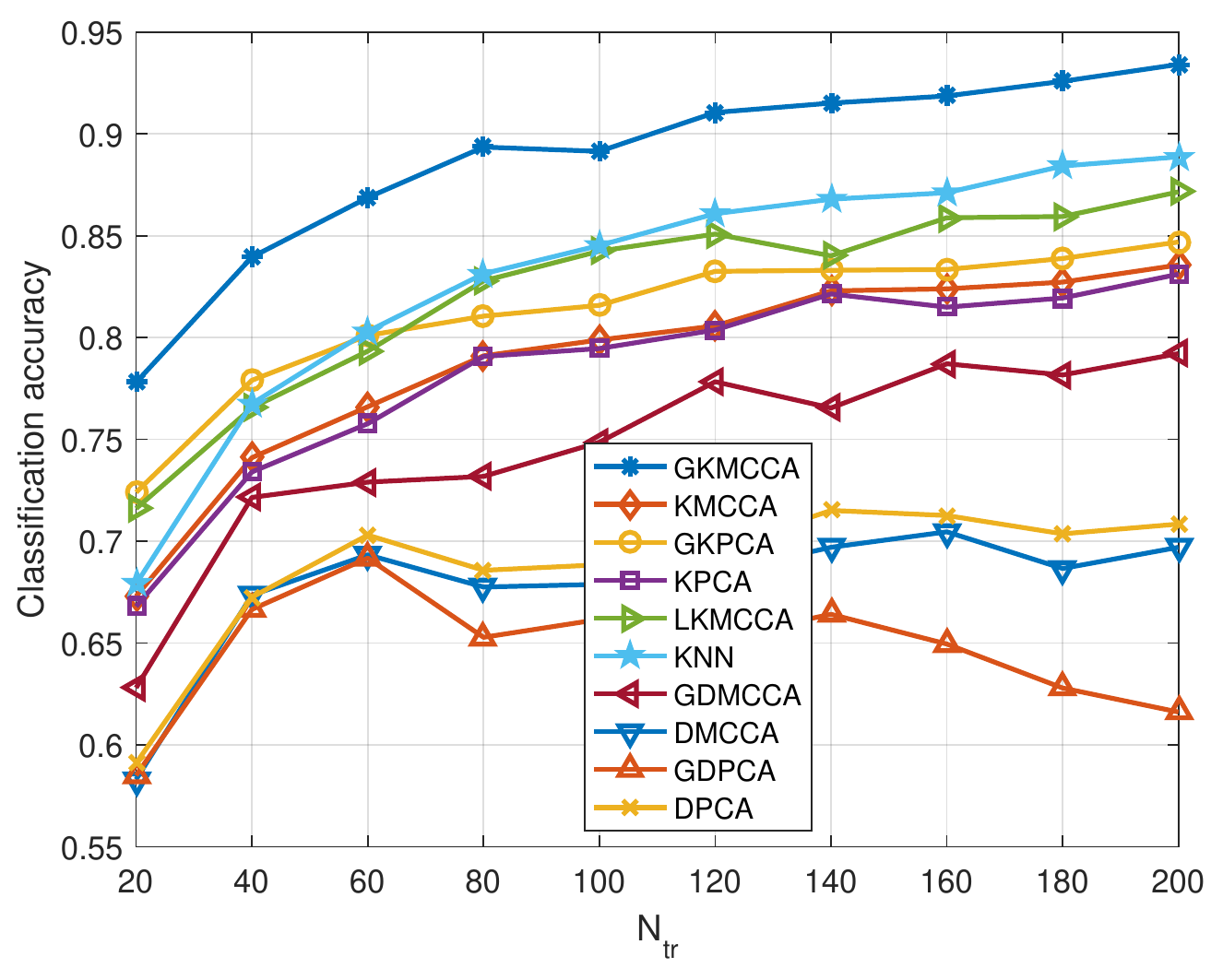}}
	\caption{Classification results using MNIST data.}
	\label{fig:mnist} 
\end{figure*}


\subsection{Image data classification}

The MNIST database\footnote{Downloaded from http://yann.lecun.com/exdb/mnist/.} containing $10$ classes of handwritten $28\times 28$ digit images with $7,000$ images per class, is used here to assess the merits of GKMCCA in classification. Per MC test, three sets of $N_{\rm tr}$ images per class were randomly picked for training, hyper-parameter tunning, and testing, respectively. We followed the process of generating the three-view training, tuning, and testing data in Sec. \ref{simu:face} to construct $\{\mathbf{X}_m\in\mathbb{R}^{196\times 10N_{\rm tr}}\}_{m=1}^3$, $\{\mathbf{X}_m^{\rm tu}\in\mathbb{R}^{196\times 10N_{\rm tr}}\}_{m=1}^3$, and $\{\mathbf{X}_m^{\rm te}\in\mathbb{R}^{196\times 10N_{\rm tr}}\}_{m=1}^3$, except that each data sample per view was resized to $14\times 14$ pixels. 

Gaussian kernels were used for $\{\mathbf{X}_m\}_{m=1}^3$, the resized, as well as the vectorized training images, denoted by $\{\mathbf{o}_i\in\mathbb{R}^{196\times 1}\}_{i=1}^{10N_{\rm tr}}$, where the bandwidth parameters were set equal to the mean of their corresponding Euclidean distances. Relying on the kernel matrix of $\{\mathbf{o}_i\}$, denoted by $\mathbf{K}_o\in\mathbb{R}^{10N_{\rm tr}\times 10 N_{\rm tr}}$, the graph adjacency matrix was constructed in the way depicted in \eqref{eq:wijdual} but with $\frac{\mathbf{o}_i^\top\mathbf{o}_j}{\|\mathbf{o}_i\|_2\|\mathbf{o}_j\|_2}$ and $20N_{\rm tr}$ replaced by the $(i,\,j)$-th entry of $\mathbf{K}_o$ and $10N_{\rm tr}$, respectively. The graph Laplacian regularized kernel multi-view (LKM) CCA \cite{blaschko2011semi} used three graph adjacency matrices, which were obtained by \eqref{eq:wijdual} after substituting $\frac{\mathbf{o}_i^\top\mathbf{o}_j}{\|\mathbf{o}_i\|_2\|\mathbf{o}_j\|_2}$ by the $(i,\,j)$-th entry of $\{\mathbf{K}_m\}_{m=1}^3$. To implement GDMCCA and GDPCA, the graph adjacency matrices were constructed via \eqref{eq:wijdual}. In all tests of this subsection, we set $k_2=N_{\rm tr}-1$. The hyper-parameters of GKMCCA, KMCCA, GKPCA, LKMCCA, GDMCCA, DMCCA, and GDPCA were selected from $30$ logarithmically spaced values between $10^{-3}$ and $10^3$, that yields the best classification performance. Ten projection vectors are learned by GKMCCA, KMCCA, GKPCA, KPCA, LKMCCA, GDMCCA, DMCCA, GDPCA, and DPCA, which are further used to obtain the low-dimensional representations of $\{\mathbf{X}_m^{\rm te}\}_{m=1}^3$. Then, the $5$-nearest neighbors rule is adopted for classification. The classification accuracies of all methods reported are averages over $30$ MC runs.

In Figs. \ref{fig:mnist_view1}, \ref{fig:mnist_view2}, and \ref{fig:mnist_view3}, the classification accuracies of the $10$-dimensional representations of $\mathbf{X}_1^{\rm te}$, $\mathbf{X}_2^{\rm te}$, and $\mathbf{X}_3^{\rm te}$ are plotted. The advantage of GKMCCA relative to other competing alternatives remains remarkable no matter which view of testing data is employed.



\section{Conclusions}
In this work, CCA along with multiview CCA was revisited. Going beyond existing (M)CCA approaches, a novel graph-regularized MCCA method was put forth that leverages prior knowledge described by graph(s) that common information bearing sources belong to. By embedding the latent common sources in a graph and invoking this extra information as a graph regularizer, our GMCCA was developed to endow the resulting low-dimensional representations. Performance analysis of our GMCCA approach was also provided through the development of a generalization bound. To cope with data vectors whose dimensionality exceeds the number of data samples, we further introduced a dual form of GMCCA. To further account for nonlinear data dependencies, we generalized GMCCA to obtain a graph-regularized kernel MCCA scheme too. Finally, we showcased the merits of our proposed GMCCA approaches using extensive real-data tests.

This work opens up several interesting directions for future research. Developing efficient GMCCA algorithms for high-dimensional multiview learning is worth investigating.  Generalizing our proposed GMCCA approaches to handle unaligned multiview datasets is also pertinent for semi-supervised learning as well. Incorporating additional structural forms regularization, e.g., sparsity and non-negativity, into the novel GMCCA framework is meaningful too.

\bibliographystyle{IEEEtranS}
\bibliography{pca}

\begin{thebibliography}{10}
\providecommand{\url}[1]{#1}
\csname url@samestyle\endcsname
\providecommand{\newblock}{\relax}
\providecommand{\bibinfo}[2]{#2}
\providecommand{\BIBentrySTDinterwordspacing}{\spaceskip=0pt\relax}
\providecommand{\BIBentryALTinterwordstretchfactor}{4}
\providecommand{\BIBentryALTinterwordspacing}{\spaceskip=\fontdimen2\font plus
\BIBentryALTinterwordstretchfactor\fontdimen3\font minus
  \fontdimen4\font\relax}
\providecommand{\BIBforeignlanguage}[2]{{%
\expandafter\ifx\csname l@#1\endcsname\relax
\typeout{** WARNING: IEEEtranS.bst: No hyphenation pattern has been}%
\typeout{** loaded for the language `#1'. Using the pattern for}%
\typeout{** the default language instead.}%
\else
\language=\csname l@#1\endcsname
\fi
#2}}
\providecommand{\BIBdecl}{\relax}
\BIBdecl

\bibitem{andrew2013deep}
G.~Andrew, R.~Arora, J.~Bilmes, and K.~Livescu, ``Deep canonical correlation
  analysis,'' in \emph{Proc. Intl. Conf. Mach. Learn.}, Atlanta, USA, June
  16-21, 2013.

\bibitem{bartlett2002rademacher}
P.~L. Bartlett and S.~Mendelson, ``Rademacher and {G}aussian complexities:
  {R}isk bounds and structural results,'' \emph{J. Mach. Learn. Res.}, vol.~3,
  pp. 463--482, Nov. 2002.

\bibitem{tweeterdata}
A.~Benton, R.~Arora, and M.~Dredze, ``Learning multiview embeddings of
  {T}witter users,'' in \emph{Proc. Annual Meeting Assoc. Comput. Linguistics},
  vol.~2, Berlin, Germany, Aug. 7-12, 2016, pp. 14--19.

\bibitem{blaschko2011semi}
M.~B. Blaschko, J.~A. Shelton, A.~Bartels, C.~H. Lampert, and A.~Gretton,
  ``Semi-supervised kernel canonical correlation analysis with application to
  human f{MRI},'' \emph{Pattern Recognit. Lett.}, vol.~32, no.~11, pp.
  1572--1583, Aug. 2011.

\bibitem{1985ace}
L.~Breiman and J.~H. Friedman, ``Estimating optimal transformations for
  multiple regression and correlation,'' \emph{J. Amer. Stat. Assoc.}, vol.~80,
  no. 391, pp. 580--598, Sep. 1985.

\bibitem{ccasumcor}
J.~D. Carroll, ``Generalization of canonical correlation analysis to three or
  more sets of variables,'' in \emph{Proc. Annual Conv. Amer. Psychol. Assoc.},
  vol.~3, 1968, pp. 227--228.

\bibitem{tsp2018cwg}
J.~Chen, G.~Wang, and G.~B. Giannakis, ``Nonlinear dimensionality reduction for
  discriminative analytics of multiple datasets,'' \emph{IEEE Trans. Signal
  Process.}, to appear Jan. 2019. [Online]. Available:
  https://arxiv.org/abs/1805.05502.

\bibitem{chen2017distributed}
J.~Chen and I.~D. Schizas, ``Distributed efficient multimodal data
  clustering,'' in \emph{Proc. European Signal Process. Conf.}, Kos Island,
  Greece, Aug. 28 - Sep. 2, 2017, pp. 2304--2308.

\bibitem{2018cwsggcca}
J.~Chen, G.~Wang, Y.~Shen, and G.~B. Giannakis, ``Canonical correlation
  analysis of datasets with a common source graph,'' \emph{IEEE Trans. Signal
  Process.}, vol.~66, no.~16, pp. 4398--4408, Aug. 2018.

\bibitem{chen2012structured}
X.~Chen, L.~Han, and J.~Carbonell, ``Structured sparse canonical correlation
  analysis,'' in \emph{Artif. Intl. Stat.}, Mar. 2012, pp. 199--207.

\bibitem{proc2018gsk}
G.~B. Giannakis, Y.~Shen, and G.~V. Karanikolas, ``Topology identification and
  learning over graphs: {A}ccounting for nonlinearities and dynamics,''
  \emph{Proc. of the IEEE}, vol. 106, no.~5, pp. 787--807, May 2018.

\bibitem{hardoon2004canonical}
D.~R. Hardoon, S.~Szedmak, and J.~Shawe-Taylor, ``Canonical correlation
  analysis: {A}n overview with application to learning methods,'' \emph{Neural
  Comput.}, vol.~16, no.~12, pp. 2639--2664, Dec. 2004.

\bibitem{1961maxvar}
P.~Horst, ``Generalized canonical correlations and their applications to
  experimental data,'' \emph{J. Clinical Psych.}, vol.~17, no.~4, pp. 331--347,
  Oct. 1961.

\bibitem{1936cca}
H.~Hotelling, ``Relations between two sets of variates,'' \emph{Biometrika},
  vol.~28, no. 3/4, pp. 321--377, Dec. 1936.

\bibitem{gpca1}
B.~Jiang, C.~Ding, and J.~Tang, ``Graph-{L}aplacian {PCA: C}losed-form solution
  and robustness,'' in \emph{Proc. Intl. Conf. Comput. Vision Pattern
  Recognit.}, Portland, USA, Jun. 25-27, 2013.

\bibitem{kanatsoulis2018structured}
C.~I. Kanatsoulis, X.~Fu, N.~D. Sidiropoulos, and M.~Hong, ``Structured sumcor
  multiview canonical correlation analysis for large-scale data,'' \emph{IEEE
  Trans. Signal Process., \emph{see also} arXiv:1804.08806}, Sept. 2018.

\bibitem{1901pca}
F.~R.~S. Karl~Pearson, ``{LIII}. {O}n lines and planes of closest fit to
  systems of points in space,'' \emph{The London, Edinburgh, and Dublin Phil.
  Mag. and J. of Science}, vol.~2, no.~11, pp. 559--572, 1901.

\bibitem{pssechap}
V.~Kekatos, G.~Wang, H.~Zhu, and G.~B. Giannakis, \emph{{PSSE} Redux: {C}onvex
  Relaxation, {D}ecentralized, {R}obust, and {D}ynamic {A}pproaches}.\hskip 1em
  plus 0.5em minus 0.4em\relax Adv. Electr. Power and Energy; Power Sys.
  Engin., M. El-Hawary Editor; \emph{see also} arXiv:1708.03981, 2017.

\bibitem{1971kettenringcca}
J.~R. Kettenring, ``Canonical analysis of several sets of variables,''
  \emph{Biometrika}, vol.~58, no.~3, pp. 433--451, Dec. 1971.

\bibitem{yaleb}
K.~C. Lee, J.~Ho, and D.~J. Kriegman, ``Acquiring linear subspaces for face
  recognition under variable lighting,'' \emph{IEEE Trans. Pattern Anal. Mach.
  Intell.}, vol.~27, no.~5, pp. 684--698, May 2005.

\bibitem{lopez2014randomized}
D.~Lopez-Paz, S.~Sra, A.~Smola, Z.~Ghahramani, and B.~Sch{\"o}lkopf,
  ``Randomized nonlinear component analysis,'' in \emph{Proc. Intl. Conf. Mach.
  Learn.}, Beijing, China, June 21-26, 2014, pp. 1359--1367.

\bibitem{2015mccanonneg}
P.~Rastogi, B.~Van~Durme, and R.~Arora, ``Multiview {L}{S}{A}: Representation
  learning via generalized {CCA},'' in \emph{Proc. North American Chap. Assoc.
  Comput. Linguistics: Human Language Tech.}, Denver, Colorado, USA, May 31-
  June 5 2015, pp. 556--566.

\bibitem{ccanphard}
J.~Rupnik, P.~Skraba, J.~Shawe-Taylor, and S.~Guettes, ``A comparison of
  relaxations of multiset cannonical correlation analysis and applications,''
  \emph{arXiv:1302.0974}, Feb. 2013.

\bibitem{jstsp2016shahid}
N.~Shahid, N.~Perraudin, V.~Kalofolias, G.~Puy, and P.~Vandergheynst, ``Fast
  robust {PCA} on graphs,'' \emph{IEEE J. Sel. Topics Signal Process.},
  vol.~10, no.~4, pp. 740--756, Feb. 2016.

\bibitem{2012gdmf}
F.~Shang, L.~Jiao, and F.~Wang, ``Graph dual regularization non-negative matrix
  factorization for co-clustering,'' \emph{Pattern Recognit.}, vol.~45, no.~6,
  pp. 2237--2250, Jun. 2012.

\bibitem{shawe2004kernel}
J.~Shawe-Taylor and N.~Cristianini, \emph{Kernel {M}ethods for {P}attern
  {A}nalysis}.\hskip 1em plus 0.5em minus 0.4em\relax 1st ed., Cambridge,
  United Kingdom: Cambridge University Press, June 2004.

\bibitem{gpca}
Y.~Shen, P.~Traganitis, and G.~B. Giannakis, ``Nonlinear dimensionality
  reduction on graphs,'' in \emph{IEEE Intl. Wksp. Comput. Adv. in Multi-Sensor
  Adaptive Process.}, Curacao, Dutch Antilles, Dec. 10-13, 2017.

\bibitem{2018scgmkl}
Y.~Shen, T.~Chen, and G.~B. Giannakis, ``Online ensemble multi-kernel learning
  adaptive to non-stationary and adversarial environments,'' in \emph{Proc. of
  Intl. Conf. on Artificial Intell. and Stat.}, Lanzarote, Canary Islands,
  April 9-11, 2018.

\bibitem{2013mlsurvey}
S.~Sun, ``A survey of multi-view machine learning,'' \emph{Neural Comput.
  App.}, vol.~23, no. 7-8, pp. 2031--2038, Dec. 2013.

\bibitem{tang2009clustering}
W.~Tang, Z.~Lu, and I.~S. Dhillon, ``Clustering with multiple graphs,'' in
  \emph{Intel. Conf. Data Mining}, Miami, Florida, USA, Dec. 6-9, 2009, pp.
  1016--1021.

\bibitem{2014smcca}
A.~Tenenhaus, C.~Philippe, V.~Guillemot, K.-A. Le~Cao, J.~Grill, and V.~Frouin,
  ``Variable selection for generalized canonical correlation analysis,''
  \emph{Biostatistics}, vol.~15, no.~3, pp. 569--583, Feb. 2014.

\bibitem{wang2015deep}
W.~Wang, R.~Arora, K.~Livescu, and J.~Bilmes, ``On deep multi-view
  representation learning,'' in \emph{Intl. Conf. Mach. Learn.}, Lille, France,
  July 6-11, 2015, pp. 1083--1092.

\bibitem{witten2009penalized}
D.~M. Witten, R.~Tibshirani, and T.~Hastie, ``A penalized matrix decomposition,
  with applications to sparse principal components and canonical correlation
  analysis,'' \emph{Biostatistics}, vol.~10, no.~3, pp. 515--534, Apr. 2009.

\bibitem{witten2009gene}
D.~M. Witten and R.~J. Tibshirani, ``Extensions of sparse canonical correlation
  analysis with applications to genomic data,'' \emph{Statis. App. Genet.
  Molecular Bio.}, vol.~8, no.~1, pp. 1--27, Jan. 2009.

\bibitem{2003kcca}
Y.~Yamanishi, J.-P. Vert, A.~Nakaya, and M.~Kanehisa, ``Extraction of
  correlated gene clusters from multiple genomic data by generalized kernel
  canonical correlation analysis,'' \emph{Bioinformatics}, vol.~19, no.~1, pp.
  i323--i330, Jul. 2003.

\bibitem{pr2014sun}
Y.~Yuan and Q.~Sun, ``Graph regularized multiset canonical correlations with
  applications to joint feature extraction,'' \emph{Pattern Recognit.},
  vol.~47, no.~12, pp. 3907--3919, Dec. 2014.

\bibitem{zhang2017going}
L.~Zhang, G.~Wang, and G.~B. Giannakis, ``Going beyond linear dependencies to
  unveil connectivity of meshed grids,'' in \emph{Proc. {IEEE} Wkshp. on
  Comput. Adv. Multi-Sensor Adaptive Process.}, Curacao, Dutch Antilles, Dec.
  2017.

\end{thebibliography}

\end{document}